\documentclass[11pt]{article}

\usepackage[margin=1in]{geometry}
\usepackage[T1]{fontenc}
\usepackage{lmodern}
\usepackage{microtype}
 
\usepackage{amsmath, amssymb, amsthm, mathtools}
\usepackage[hidelinks]{hyperref}
\usepackage[nameinlink,capitalise,noabbrev]{cleveref}
\usepackage{xcolor}
\usepackage{tikz}
\usetikzlibrary{shapes,decorations.pathmorphing,calc,patterns,shadings}
\usepackage{pathslib}
\usepackage{algorithm}
\usepackage[noend]{algpseudocode}

\input{flowerlib.sty}

\allowdisplaybreaks[1]

\theoremstyle{plain}
\newtheorem{theorem}{Theorem}
\newtheorem*{theorem*}{Theorem}
\newtheorem{lemma}{Lemma}
\newtheorem{claim}{Claim}

\newtheorem{corollary}{Corollary}

\newtheorem{definition}{Definition}
\newtheorem{remark}{Remark}

\theoremstyle{definition}

\theoremstyle{remark}

\newcommand{\avec}{\vec{\alpha}}
\newcommand{\Gnhalf}{\mathbb{G}(n,1/2)}
\newcommand{\greedy}{\textsc{Greedy} }
\newcommand{\done}{\textsf{AddedPair}}
\newcommand{\T}{\textsf{True}}
\newcommand{\F}{\textsf{False}}
\newcommand{\E}{\mathbb{E}}
\newcommand{\ER}{Erd\H{o}s--R\'{e}nyi }
\newcommand{\N}{\mathbb{N}}

\newcommand{\eps}{\varepsilon}
\newcommand{\MCS}{\mathsf{LCIS}}

\newcommand{\bits}{\{0,1\}}

\newcommand{\Alg}{\mathcal{A}}
\newcommand{\findOneLarge}{\mathcal{E}}
\newcommand{\findManyLarge}{\mathcal{S}}
\newcommand{\online}{\mathsf{ONLINE}}
\newcommand{\lineref}[1]

\usepackage[
  backend=biber,
  style=alphabetic,   %
  maxbibnames=99,
]{biblatex}
\renewbibmacro{in:}{}

\renewcommand{\Pr}{\mathbb{P}}

\title{Optimal Hardness of Online Algorithms\\ for Large Common Induced Subgraphs}
\author{David Gamarnik\thanks{MIT; \url{gamarnik@mit.edu}}
\and Mikl\'os Z. R\'acz\thanks{Northwestern University; \url{miklos.racz@northwestern.edu}} 
\and Gabe Schoenbach\thanks{University of Chicago; \url{gschoenbach@uchicago.edu}}}
\date{\today}

\begin{document}
\maketitle

\begin{abstract}
We study the problem of efficiently finding large common induced subgraphs of two independent Erd\H{o}s--R\'enyi random graphs $G_1, G_2 \sim \mathbb{G}(n,1/2)$. Recently, Chatterjee and Diaconis~\cite{CD23} showed that the largest common induced subgraph of $G_1$ and $G_2$ has size $(4-o(1))\log_2 n$ with high probability. We first show that a simple greedy online algorithm finds a common induced subgraph of $G_1$ and $G_2$ of size $(2-o(1)) \log_2 n$ with high probability. Our main result shows that no online algorithm can find a common induced subgraph of $G_1$ and $G_2$ of size at least $(2+\varepsilon) \log_2 n$ with probability bounded away from $0$ as $n \to \infty$. Together, these results provide evidence that this problem exhibits a computation-to-optimization gap. To prove the impossibility result, we show that the solution space of the problem exhibits a version of the (multi) overlap gap property (OGP), and utilize an interpolation argument recently developed by Gamarnik, K{\i}z{\i}lda\u{g}, and Warnke~\cite{GKW25} that connects OGP and online algorithms.  
\end{abstract}

\section{Introduction}
In the \emph{largest common induced subgraph} ($\MCS$) problem, the input is a pair of graphs $G_{1}$ and~$G_{2}$, 
and the goal is to find a common induced subgraph of $G_{1}$ and $G_{2}$ with as many vertices as possible. 
A \emph{solution} $\sigma = (H_1, H_2)$ is any pair of induced subgraphs $H_1 \subseteq G_1$ and $H_2 \subseteq G_2$ such that $H_1 \cong H_2$, where $\cong$ denotes graph isomorphism. The \emph{size} of the solution $|\sigma| := |H_1| = |H_2|$ is the number of vertices of the common subgraph and the goal is to maximize this quantity. 

The $\MCS$ problem is a fundamental optimization problem with close connections to other central problems such as graph isomorphism and graph matching, 
and it arises naturally in many applications, 
such as bioinformatics~\cite{ehrlich2011maximum}, chemoinformatics~\cite{schietgat2013polynomial}, and pattern recognition~\cite{conte2004thirty}. 
It is computationally hard in the worst case, even to approximate within any reasonable factor~\cite{kann1992approximability}.

In this paper we study arguably the simplest and most natural average-case version of the $\MCS$ problem, 
where 
$G_{1}, G_{2} \sim \Gnhalf$ 
are independent \ER random graphs with $n$ vertices and edge probability $1/2$; we refer to this problem as $\MCS(n)$. 
Recently, Chatterjee and Diaconis~\cite{CD23} showed that the largest solution has size $(4 - o(1))\log_2 n$ with high probability as $n \to \infty$.\footnote{In fact, they proved a much more precise two-point concentration result; however, in this paper we primarily focus on the leading order term.}\textsuperscript{,}\footnote{Throughout the paper, we say that an event holds with high probability if its probability tends to $1$ as $n \to \infty$.} 

In this work, we initiate the study of the algorithmic side of the $\MCS(n)$ problem. We ask: %

\begin{center}
    \emph{What are the limits of efficient algorithms that solve $\MCS(n)$ with high probability?}
\end{center}

In particular, we focus on a broad class of online algorithms, which arise naturally for this problem. 
We first show that a simple greedy online algorithm finds a solution of size at least $(2-o(1)) \log_2 n$ with high probability. In other words, this efficient algorithm is half-optimal and serves as a natural benchmark.

\begin{theorem}[Algorithms, informal]
    There exists an efficient online algorithm that finds a solution to $\MCS(n)$ that 
    has size at least $(2 - o(1)) \log_2 n$ with high probability as $n \to \infty$.
\end{theorem}

Our main result shows that no online algorithm can do better than half-optimal. 
The class of online algorithms we consider is broad; roughly, it is defined as follows. 
An online algorithm for $\MCS$ processes the two graphs $G_{1}$ and $G_{2}$ one vertex pair $(u_{t}, v_{t}) \in V(G_{1}) \times V(G_{2})$ at a time. When a vertex pair is processed, all connections of $u_{t}$ and $v_{t}$ to previously processed vertices are revealed (in $G_{1}$ and $G_{2}$, respectively). The algorithm then has to make a decision whether or not to add a pair of vertices to the partial solution it has built so far. The main constraints are that both of the added vertices must have been already processed and \emph{at least} one of them must be $u_{t}$ or~$v_{t}$. 

Notably, an online algorithm does \emph{not} have to add both $u_{t}$ and $v_{t}$ at the same time, and as a result, it is possible for an individual vertex to be added to the solution well after it was originally processed. 
This is a fundamental difference compared to online algorithms for problems involving a single graph, where usually an immediate decision must be made about a vertex as soon as it is processed.\footnote{We note that in the problem of finding large cliques in the ``one-dimensional'' single graph case considered earlier in~\cite{GKW25}, postponing the decision about including/excluding $u_t$ trivializes the problem. Indeed, one can delay all decisions until all but $(2-o(1))\log_2 n$ vertices are revealed, find a nearly largest clique among the revealed vertices, and then include the vertices of the discovered clique in the last $(2-o(1))\log_2 n$ rounds.}
This difference stems from the ``two-dimensional'' nature of the $\MCS$ problem, with the input being two graphs. 
We defer formal definitions to Section~\ref{sec:impossibility_results} below.

\begin{theorem}[Impossibility, informal]\label{thm:impossibility-informal}
    Fix any $\eps > 0$. There is no online algorithm for $\MCS(n)$ that finds a solution of size at least $(2+\eps) \log_{2} n$ with probability bounded away from $0$ as $n \to \infty$. 
\end{theorem}

Overall, these results combined provide evidence that the $\MCS(n)$ problem exhibits a computation-to-optimization gap. 
We next turn to describing our results more formally and in more detail.

\subsection{A half-optimal greedy online algorithm}
A naive brute-force approach to finding a size-$k$ solution to $\MCS(n)$ would require checking all pairs of size-$k$ subgraphs of $G_1$ and $G_2$. For $k = \Theta(\log n)$, such an exhaustive approach would require superpolynomial runtime, since
$\binom{n}{k}^2 \geq (n/k)^{2k} \geq n^{\Omega(\log n)}$. 

To get a polynomial-time algorithm, a useful initial observation is that independent sets of equal size are isomorphic, so any algorithm that finds independent sets can also solve the $\MCS$ problem (and the same applies to cliques). 
The largest independent set in $G \sim \Gnhalf$ has size $(2-o(1)) \log_2 n$ with high probability 
and a simple greedy online algorithm finds an independent set of size $(1+o(1)) \log_2 n$~\cite{grimmett1975colouring,Kar76}. 
So this immediately gives an efficient quarter-optimal solution to $\MCS(n)$: use a greedy algorithm to find independent sets of size $(1+o(1)) \log_2 n$ in both $G_{1}$ and $G_{2}$. 
However, there is no known efficient algorithm to find an independent set of size $(1+\eps) \log_2 n$ in $G \sim \Gnhalf$ --- indeed, it is conjectured that such an algorithm does not exist --- so it is not immediately clear how to go beyond~this.

The key idea is to tailor the greedy algorithm for finding an independent set to the $\MCS$ problem. 
Recall first the algorithm for finding an independent set in a single graph $G$~\cite{grimmett1975colouring,Kar76}, 
which processes the vertices $u_{1}, \ldots, u_{n}$ of $G$ in an arbitrary order. It initializes the solution set $S$ to be empty and proceeds greedily, adding $u_{t}$ to $S$ if and only if $u_{t}$ is not connected to any of the vertices in~$S$. 
After processing all vertices, it outputs the induced subgraph $G[S]$, which is an independent set. 

Turning to the $\MCS$ problem, the input is now two graphs $G_{1}$ and $G_{2}$. 
Let $u_{1}, \ldots, u_{n}$ and 
$v_{1}, \ldots, v_{n}$ be arbitrary vertex orderings of $G_{1}$ and $G_{2}$, respectively, 
and initialize the pair of induced subgraphs $H_{1}$ and $H_{2}$ to be empty. 
The greedy algorithm processes vertex pairs $(u_{t}, v_{t})$ one at a time and does the following at time $t$. 
First, all connections of $u_{t}$ and $v_{t}$ to previously processed vertices are revealed (in $G_{1}$ and $G_{2}$, respectively). 
Then, for all $j \leq t$, do the following: 
\begin{itemize}
    \item if $u_t$ and $v_j$ can be added to the subgraphs $H_1$ and $H_2$, respectively, while preserving $H_1 \cong H_2$, then do so; 
    \item otherwise, perform the same check for $u_{j}$ and $v_{t}$.
\end{itemize}
Note that once a vertex pair is added to $H_{1}$ and $H_{2}$ at time $t$, 
the algorithm moves on to time~$t+1$. 
Once all vertices are processed, the algorithm returns $(H_{1}, H_{2})$ and a bijection $\pi : V(H_{1}) \to V(H_{2})$ 
that verifies $H_{1} \cong H_{2}$.\footnote{Note that outputting a graph isomorphism $\pi$ that verifies $H_{1} \cong H_{2}$ is not strictly necessary to solve the $\MCS$ problem.}
See Algorithm~\ref{alg:greedy} in Section~\ref{sec:greedy-proof} for a detailed description of the algorithm.

\begin{theorem}[Algorithms, formal]\label{thm:algs-formal}
    Let $(G_1, G_2) \sim \Gnhalf^{\otimes 2}$, define $\greedy$  as in Algorithm~\ref{alg:greedy}, 
    and let $(H_1, H_2, \pi)$ denote the output of $\greedy(G_1, G_2)$. 
    Then $(H_1, H_2)$ is a solution to $\MCS(n)$ and $\pi$ is a graph isomorphism between $H_{1}$ and $H_{2}$. 
    Furthermore, for any $\eps > 0$ and for all $n$ large enough, the size of $(H_1, H_2)$ is at least $(2 - \eps)\log_2 n$ with probability at least $1 - \exp(-n^{\eps / 4})$.
\end{theorem}

The \greedy algorithm has runtime $O(n^{2} \log n)$ (i.e., quasilinear in the input size) with high probability in the average-case setting of Theorem~\ref{thm:algs-formal}, and $O(n^{3})$ in the worst case (see Claim~\ref{clm:runtime}). 
Also, the proof of Theorem~\ref{thm:algs-formal} shows that the conclusion still holds for $\eps = \eps_{n} \geq \frac{C \log \log n}{\log n}$ for large enough constant $C$, 
and in particular it shows that the output of \greedy has size at least $2 \log_2 n - 9 \log \log n$ with high probability (see Corollary~\ref{cor:greedy}).

\subsection{The limits of online algorithms}\label{sec:impossibility_results}

Our impossibility result rules out a certain broad class of online algorithms, which includes the \greedy algorithm described in Algorithm~\ref{alg:greedy}. Instead of receiving the input all at once, these algorithms receive bits of the input in serial and \emph{make irrevocable decisions about the output at each step}. It is this irrevocability property that constrains online algorithms; otherwise, the serial processing of the input is moot. Unlike the closely related independent set problem on a single graph, the input to $\MCS$ consists of \emph{two} graphs, so there are several distinct notions of how information about the pair $(G_1, G_2)$ should be processed by an online algorithm. What ought to be the atomic object? 

We propose the following natural definition. At a high level, at each step, an online algorithm $\Alg$ will process a \emph{pair} of vertices $(v_1, v_2) \in V(G_1) \times V(G_2)$ and decide whether or not to add a (potentially distinct) pair of vertices to the partial solution, under the constraint that both of the added vertices must have been already processed, and \emph{at least} one of them must be $v_1$ or $v_2$. As a result, it is possible for an individual vertex to be added to the solution after it was processed. As highlighted previously, this is a fundamental difference compared to the definition of online algorithms for independent set~\cite{GKW25}, and is due to the two-dimensional nature of the $\MCS$ problem. Our formal definition is below.

\begin{definition}[Online algorithms for $\MCS$]\label{def:online}
    The class $\online(n)$ consists of algorithms $\Alg$ which take as input $(G_1, G_2)$ and a random string $\omega \in \Omega$, where for each $j \in \{1, 2\}$, $G_j = (V_j, E_j)$ is a graph with $V_j = [n]$. Each $\Alg \in \online(n)$ proceeds in rounds, and at the end of round $t \in [n]$, $\Alg$ maintains the following~state:
    \begin{itemize}
        \item The sets of vertices $P^{(t)}_1 \subseteq V_1$ and $P^{(t)}_2 \subseteq V_2$ inspected thus far.
        \item Vertex sets $S^{(t)}_1 \subseteq P^{(t)}_1$ and $S^{(t)}_2 \subseteq P^{(t)}_2$ that induce a pair of isomorphic subgraphs. We write
        \begin{align}
            \mathcal{A}^{(t)}(G_{1}, G_{2}, \omega) := \left(G_1\left[S^{(t)}_1\right], G_2\left[S^{(t)}_2\right]\right)
        \end{align}
        for the (partial) solution at the end of round $t \in [n]$. 
    \end{itemize}
    $\Alg$ initializes $P^{(0)}_1 = P^{(0)}_2 = S^{(0)}_1 = S^{(0)}_2 = \emptyset$. During each round $t \in [n]$, $\Alg$ proceeds as follows, where for all $\{u, v\} \in \binom{[n]}{2}$, $e_j(\{u, v\}) \in \bits$ denotes the edge status of the vertex pair $\{u, v\}$ in graph~$G_j$.
    \begin{enumerate}
        \item \emph{(Processing a new vertex pair.)} 
        Based on $\omega$ and all of the information revealed through the end of round $t-1$ 
        (namely, $P_{1}^{(t-1)}$, $P_{2}^{(t-1)}$, $G_{1}\left[ P_{1}^{(t-1)} \right]$, $G_{2}\left[ P_{2}^{(t-1)} \right]$, $S_{1}^{(t-1)}$, and $S_{2}^{(t-1)}$), 
        for $j \in \{1, 2\}$, 
        the algorithm $\Alg$ selects $v^{(t)}_j \in V_j \setminus P^{(t-1)}_j$ and defines $P^{(t)}_j := P^{(t-1)}_j \cup \{v^{(t)}_j\}$. It then reveals $e_j(\{u, v\})$ for all not-yet-seen vertex pairs $\{u, v\} \in \binom{P^{(t)}_{j}}{2}$, for $j \in \{1,2\}$.
        \item \emph{(Updating the solution.)} 
        Based on $\omega$ and all of the information revealed thus far, $\Alg$ chooses how to update the vertex sets as
        \begin{align}
            \left(S^{(t)}_1, S^{(t)}_2\right) \in \begin{cases}
                \left(S^{(t-1)}_1, S^{(t-1)}_2\right) \\
                \left(S^{(t-1)}_1 \cup \{v_1\}, S^{(t-1)}_2 \cup \{v_2\}\right),
            \end{cases}
        \end{align}
        under the constraint that $v_j \in P^{(t)}_j$ for $j \in \{1,2\}$, and 
        either $v_{1} = v_{1}^{(t)}$ 
        or $v_{2} = v_{2}^{(t)}$ 
        (or both). 
        Moreover, the algorithm $\Alg$ must maintain that $G_1\left[S^{(t)}_1\right] \cong G_2\left[S^{(t)}_2\right]$.
    \end{enumerate}
    We write the output of $\Alg$ as
    \begin{align}
        \Alg(G_{1}, G_{2}, \omega) := \Alg^{(n)}(G_{1}, G_{2}, \omega) = \left(G_1\left[S^{(n)}_1\right], G_2\left[S^{(n)}_2\right]\right).
    \end{align}
    We write $|\Alg(G_{1}, G_{2}, \omega)| := \left|S^{(n)}_1\right| = \left|S^{(n)}_2\right|$ to denote the size of the output. 
\end{definition}

Aside from the two-dimensional nature of $\MCS$, Definition~\ref{def:online} closely tracks \cite[Definition~2.1]{GKW25}. One difference is that the latter definition includes algorithms that can choose a set of \emph{exceptional edges}, which allows for better-informed decisions at each step. We believe that Definition~\ref{def:online} could be similarly strengthened to include exceptional information without affecting our results (assuming an appropriate bound on the amount of exceptional information), but we choose not to for simplicity.

Our main result is the following, which rules out online algorithms beyond $2 \log_2 n$.

\begin{theorem}[Impossibility, formal]\label{thm:impossibility-formal}
    Let $(G_1, G_2) \sim \Gnhalf^{\otimes 2}$, let $\Alg \in \online(n)$, and fix any $\eps > 0$. For all $\omega \in \Omega$ and all sufficiently large $n$, we have that 
    \begin{align}
        \Pr\left(|\Alg(G_1, G_2, \omega)| \geq (2 + \eps)\log_2 n \right) \leq n^{-\Omega(\log n)}.
    \end{align} 
\end{theorem}

\subsection{Related work}

\paragraph{Existential results for $\MCS$.}
The $\MCS(n)$ problem was first studied by Chatterjee and Diaconis~\cite{CD23}, who showed a precise two-point concentration result: 
with high probability the largest common induced subgraph has size 
$\left\lfloor x_{n} - \eps \right\rfloor$ 
or 
$\left\lfloor x_{n} + \eps \right\rfloor$, 
where 
$x_{n}$ satisfies
$x_{n} = (4-o(1))\log_2 n$. 
Subsequently, Surya, Warnke, and Zhu~\cite{SWZ25} proved a more general result
for arbitrary edge densities $p_1, p_2 \in (0,1)^2$ in the two graphs, 
with three qualitatively different regimes of $p_1, p_2 \in (0,1)^2$. 
Diamantidis, Konstantopoulos, and Yuan~\cite{DKY25} used a different proof technique to show the same result, albeit for a slightly restricted set of edge densities.

\paragraph{Independent set.}
$\MCS$ is closely related to the much-studied problem of finding large independent sets. In the past decade a line of work has focused on giving evidence of the computational hardness of finding large independent sets in sparse random graphs $G \sim \mathbb{G}(n,d/n)$. Gamarnik and Sudan~\cite{GS17} showed that no local algorithms can find independent sets with size above some multiplicative factor $\beta$ of optimal, for $\beta = 1/2 + 1/(2\sqrt{2}) \approx 0.85$. Rahman and Vir\'ag~\cite{RV17} strengthened this result to show that local algorithms cannot produce independent sets larger than half-optimal size, which is tight. Soon after, Wein~\cite{Wei21} (building off of \cite{GJW24}) strengthened this result by extending the class of implicated algorithms to stable ones, which includes low-degree polynomials. Observing that a natural half-optimal approach \cite{grimmett1975colouring,Kar76} is online (but not stable), Gamarnik, K{\i}z{\i}lda\u{g}, and Warnke~\cite{GKW25} proved that a natural class of online algorithms cannot beat half-optimality. 

\paragraph{Overlap gap properties.}
All of the computational hardness results above rely on proving that the independent set problem exhibits a type of \emph{overlap gap property} (OGP). At a high level, an OGP is a statement about the geometry of the solution space; it stipulates the non-existence of tuples of large solutions with a particular overlap pattern. Proving more sophisticated OGPs (for instance, allowing solutions to come from an ``ensemble'' of problem instances rather than a single one, or considering $m$-tuples of solutions rather than pairs) have allowed for tighter hardness results. For instance, Wein~\cite{Wei21} constructs an ``interpolation path'' of correlated random graphs $G_1, \dots, G_T$. The author argues by contradiction that, given a stable algorithm $\Alg$ that outputs large solutions, one could run $\Alg$ on each $G_i$ and construct an $m$-tuple of solutions that violates an ensemble $m$-OGP.

Recently, Gamarnik, K{\i}z{\i}lda\u{g}, and Warnke~\cite{GKW25} introduced a novel interpolation argument, building an interpolation path of correlated inputs that are defined in reference to any algorithm $\Alg$ in a class of online algorithms that they define. This differs from the above OGP hardness results in that the non-existent structure is not of a tuple of solutions to the problem, but rather of a tuple of solutions \emph{created by a particular type of algorithm}. As a result, the hardness proof is coextensive with the proof of the OGP, rather than coming afterward in, say, a proof by contradiction. 
This approach has several advantages, including that it establishes \emph{strong hardness} (i.e., ruling out algorithms succeeding even with only $o(1)$ probability), something that was elusive in OGP-style arguments prior to the work of Huang and Sellke~\cite{HS25}. 
Our impossibility proof utilizes this novel interpolation argument by~\cite{GKW25}; our main new ingredient is to prove an OGP for $\MCS(n)$.

\paragraph{Computational hardness of random problems.}
Our work fits into a broader landscape of using the overlap gap property to prove hardness results for random optimization problems. See \cite{DGH25, GKPX23, BGG26} for further impossibility results for online algorithms. Other hardness results include \cite{GKPX22, huang2025tight, HS25}. 
We note that such computation-to-optimization gaps have fundamental connections with statistical-computation gaps in planted problems which are widespread in high-dimensional learning. 
We do not attempt to survey the literature here --- see \cite{Gam21, Gam25} for recent surveys and more.

\subsection{Open problems}
We leave open several natural questions for future work. 

    \paragraph{General dense case.} In this work we assumed that both graphs have edge density $1/2$. For arbitrary edge densities $(p_1, p_2) \in (0,1)^2$, Surya, Warnke, and Zhu~\cite{SWZ25} 
    determined that the size of the largest common induced subgraph is roughly $C_{p_1, p_2} \log_{2} n$, where interestingly the constant $C_{p_1, p_2}$ takes on qualitatively different forms in three different regimes of $(p_{1}, p_{2}) \in (0,1)^{2}$. 
    What can we say about efficient algorithms in this more general setting? Is greedy still half-optimal, and does the impossibility proof extend naturally to this setting?

    \paragraph{Sparse case.} Can we find evidence of computational hardness for $\MCS$ given $G_1, G_2 \sim \mathbb{G}(n,d/n)$? Do these mirror analogous results for large independent sets in sparse graphs?
    
    \paragraph{Other classes of algorithms.} Can we rule out other classes of algorithms? For instance, it would be natural to consider online algorithms with exceptional edge-pairs, as well as stable algorithms.

    \paragraph{Many graphs.} What can we say about $\MCS$ for many graphs? A quick calculation suggests that a greedy algorithm remains half-optimal; it would be interesting to explore this setting further. %

    \paragraph{Average-case reduction from finding large independent sets.} 
    As discussed above, the results that we obtained for $\MCS(n)$ have similarities to those for finding large independent sets in $G \sim \Gnhalf$; in particular, a greedy algorithm is half-optimal in both settings, and conjecturally there is no efficient algorithm that does better than half-optimal. 
    Is there a deeper connection between these two problems? In particular, is there an average-case reduction from finding large independent sets to finding large common induced subgraphs (or vice versa)? 

    One attempt could be to consider independent sets in $G \sim \mathbb{G}(n^{2}, 1/2)$, noting that the largest independent set in $G$ has size roughly $2 \log_{2} (n^{2}) = 4 \log_2 n$, 
    which is the same as the largest common induced subgraph in $(G_{1}, G_{2}) \sim \Gnhalf^{\otimes 2}$. Moreover, the $n^{2}$ vertices of $G$ have a natural correspondence with ordered pairs of vertices $(u,v) \in V(G_{1}) \times V(G_{2})$. 
    It is not clear how to use these ideas to build a formal average-case reduction; doing so would be very interesting.

\subsection{Notation}

We write $G \cong H$ to denote that two graphs $G$ and $H$ are isomorphic. 
For a graph $G$ and a subset $S \subseteq V(G)$, we let $G[S]$ denote the subgraph of $G$ induced by $S$. We let $[n] := \{1, \ldots, n \}$. The size $|\sigma|$ of a solution $\sigma = (H_1, H_2)$ to the $\MCS$ problem is the number of vertices in each subgraph $|\sigma| = |V(H_1)| = |V(H_2)|$. 
We let $\mathbf{1}\{\mathcal{E}\}$ denote the indicator random variable of an event $\mathcal{E}$. 
We use standard asymptotic notation throughout. 

\subsection{Outline}

The rest of the paper is organized as follows. 
In Section~\ref{sec:proof-ideas} we describe the high level proof ideas of Theorems~\ref{thm:algs-formal} and~\ref{thm:impossibility-formal}, focusing on the more involved proof of Theorem~\ref{thm:impossibility-formal}. 
Section~\ref{sec:greedy-proof} contains the analysis of the greedy algorithm and proves Theorem~\ref{thm:algs-formal}. 
Finally, Section~\ref{sec:impossibility-proof} contains the proof of the impossibility result, Theorem~\ref{thm:impossibility-formal}. 

\section{Proof ideas}\label{sec:proof-ideas}

\subsection{Greedy algorithm}\label{sec:greedy-ideas}

In Section~\ref{sec:greedy-proof} we prove Theorem~\ref{thm:algs-formal}, that the algorithm $\greedy$ from Algorithm~\ref{alg:greedy} produces half-optimal solutions to $\MCS(n)$. Recall that $\greedy$ builds up the common induced subgraph iteratively: at step $t \in [n]$, $\greedy$ processes the vertex pair $(u_t, v_t) \in V(G_1) \times V(G_2)$. If there exists some already-processed vertex $v_j$ (i.e., with $j \leq t$) such that $(u_t, v_j)$ can be added to the partial solution, it does so --- and similarly for any $(u_j, v_t)$. 

Our goal is to give an upper bound on the probability that $\greedy$ outputs a common induced subgraph of size less than $(2 - \eps)\log_2 n$. We observe that if this happens, then there must be a window of $\widetilde{n}$ successive steps where $\greedy$ does not add to the partial solution, for
\begin{align*}
    \widetilde{n} := \left\lfloor \frac{n}{2\log_2 n} \right\rfloor.
\end{align*}
In other words, over the course of the window, $\greedy$ considers and rejects at least $\widetilde{n}^2$ vertex pairs~$(u, v) \in V(G_{1}) \times V(G_{2})$. Let $1 \leq k < (2-\eps) \log_2 n$ be the size of the partial solution at the start of the window --- we will eventually take a union bound over all possible values of $k$. If each pair $(u,v)$ were independent, the probability of rejecting all such pairs would be at most
\begin{align*}
    \left(1 - \frac{1}{2^k}\right)^{\widetilde{n}^2} \leq \exp \left( - \frac{\widetilde{n}^2}{2^k}\right) < \exp(-\widetilde{\Omega}(n^\eps)),
\end{align*}
which is small as desired. 
Of course, the vertex pairs are not all independent; we handle this via a careful conditioning argument.

\subsection{The limits of online algorithms}\label{sec:impossibility-ideas}

In Section~\ref{sec:impossibility-proof} we prove Theorem~\ref{thm:impossibility-formal}, that there is a vanishingly small probability that any $\Alg \in \online(n)$ can produce a solution to $\MCS$ that is better than half-optimal. We show in Lemma~\ref{lem:exists-deterministic} that it suffices to consider deterministic algorithms. At a high level, we relate the probability that $\Alg$ finds a large solution to the probability that $\Alg$ finds a large solution within every input from a carefully-constructed set of correlated inputs, which will ensure that the solutions have substantial mutual overlap. Finally, we will use a combinatorial argument and the first moment method to show that such a set of solutions exists with vanishingly small probability.

In more detail, we say that a large solution $\sigma$ to the $\MCS(n)$ problem is one of size $|\sigma| \geq (2 + \eps)\log_2 n$. We consider the event $\findOneLarge$, that an arbitrary deterministic algorithm $\Alg$ finds a large solution in an input $Y = (G_1, G_2) \sim \Gnhalf^{\otimes 2}$. We start by relating the probability of $\findOneLarge$ to the probability of the event $\findManyLarge$, that $\Alg$ finds a set $\{\sigma_1, \dots, \sigma_m\}$ of large solutions from a set of inputs called an \emph{interpolation path} (Figure~\ref{fig:interpolation-paths} and Definition~\ref{def:interpolation-paths}). An interpolation path $Y_1^{(t)}, \dots, Y_m^{(t)}$ is constructed so that on each input $Y_i^{(t)}$, $\Alg$ behaves identically for the first $t \in [n]$ steps of the process, and thereafter independently, conditioned on the information learned in the first $t$ steps. Moreover, the marginal distribution of each input is $\Gnhalf^{\otimes 2}$. 

In Lemma~\ref{lem:S-below}, we show that for any constant $m \in \N$, we have that 
\begin{align*}
    \Pr(\findOneLarge)^m \leq \Pr(\findManyLarge).
\end{align*}
The proof uses the conditional independence of each $\Alg\left(Y_i^{(t)}\right)$ after $t$ steps and an application of Jensen's inequality.

\begin{figure}[t!]%
\centering
\begin{tikzpicture}[scale=0.65]
    
    \drawpair{-1.3*\columnspacing}{0}{gray5}{gray5}{}{}{\Large $Y$};
    
    \draw[-{Stealth[length=3mm, width=2mm]}, thick] 
        (-1.3*\columnspacing + 2*\circleradius + \pairspacing/2 + 5pt + 3pt, 0.5cm-7.5pt) 
        -- (0 - 2*\circleradius - \pairspacing/2 - 5pt - 3pt, 2*\rowspacing - 0.5cm-7.5pt);
    \draw[-{Stealth[length=3mm, width=2mm]}, thick] 
        (-1.3*\columnspacing + 2*\circleradius + \pairspacing/2 + 5pt + 3pt, -7.5pt) 
        -- (0 - 2*\circleradius - \pairspacing/2 - 5pt - 3pt, -7.5pt);
    \draw[-{Stealth[length=3mm, width=2mm]}, thick] 
        (-1.3*\columnspacing + 2*\circleradius + \pairspacing/2 + 5pt + 3pt, -0.5cm-7.5pt) 
        -- (0 - 2*\circleradius - \pairspacing/2 - 5pt - 3pt, -2.5*\rowspacing + 0.5cm-7.5pt);
    
    \drawpair{0}{2*\rowspacing}{gray5}{gray5}{
        \filldraw[greenshape, draw=black, thin] \leftcircle;
    }{
        \filldraw[greenshape, draw=black, thin] \rightcircle;
    }{\large $Y_1^{(1)}$};
    \drawpair{\columnspacing}{2*\rowspacing}{gray2}{gray7}{
        \filldraw[greenshape, draw=black, thin] \leftcircle;
    }{
        \filldraw[greenshape, draw=black, thin] \rightcircle;
    }{\large $Y_2^{(1)}$};
    \drawpair{2.5*\columnspacing}{2*\rowspacing}{gray6}{gray2}{
        \filldraw[greenshape, draw=black, thin] \leftcircle;
    }{
        \filldraw[greenshape, draw=black, thin] \rightcircle;
    }{\large $Y_m^{(1)}$};
    
    \node at (0, 0.9*\rowspacing) {\Huge $\shortstack{.\\[1pt].\\[1pt].}$};

    \node at (1.75*\columnspacing, 2*\rowspacing-7.5pt) {\Huge $\cdot\hspace{3.5pt}\cdot\hspace{3.5pt}\cdot$};
    
    \drawpair{0}{0}{gray5}{gray5}{%
        \filldraw[greenshape, draw=black, thin] \leftgreenshape;
    }{%
        \filldraw[greenshape, draw=black, thin] \rightgreenshape;
    }{\large $Y_1^{(t)}$};
    \drawpair{\columnspacing}{0}{gray8}{gray3}{%
        \filldraw[greenshape, draw=black, thin] \leftgreenshape;
    }{%
        \filldraw[greenshape, draw=black, thin] \rightgreenshape;
    }{\large $Y_2^{(t)}$};
    
    \node at (1.75*\columnspacing, -7.5pt) {\Huge $\cdot\hspace{3.5pt}\cdot\hspace{3.5pt}\cdot$};
    
    \drawpair{2.5*\columnspacing}{0}{gray4}{gray9}{%
        \filldraw[greenshape, draw=black, thin] \leftgreenshape;
    }{%
        \filldraw[greenshape, draw=black, thin] \rightgreenshape;
    }{\large $Y_m^{(t)}$};
    
    \node at (0, -1.35*\rowspacing) {\Huge $\shortstack{.\\[3.5pt].\\[3.5pt].}$};
    
    \drawpair{0}{-2.5*\rowspacing}{greenshape}{greenshape}{}{}{\large $Y_1^{(n)}$};
    \drawpair{\columnspacing}{-2.5*\rowspacing}{greenshape}{greenshape}{}{}{\large $Y_2^{(n)}$};
    \drawpair{2.5*\columnspacing}{-2.5*\rowspacing}{greenshape}{greenshape}{}{}{\large $Y_m^{(n)}$};

    \node at (1.75*\columnspacing, -2.5*\rowspacing-7.5pt) {\Huge $\cdot\hspace{3.5pt}\cdot\hspace{3.5pt}\cdot$};

\end{tikzpicture}
\caption{A family of interpolation paths based on $\Alg \in \online(n)$, $Y \sim \Gnhalf^{\otimes 2}$, and $m \in \N$. For each $t \in [n]$, row $t$ is an interpolation path $Y_1^{(t)}, \dots, Y_m^{(t)}$, with $Y_1^{(t)} := Y$. For graph $j \in \{1,2\}$ in each input, the green region corresponds to the vertices $P_j^{(t)}$ that have been processed by $\Alg(Y)$ after $t$ steps. The edge-statuses of all pairs of nodes in the green regions are fixed for all graphs in the family. Aside from the first column (which are all~$Y$), the edge status of all other pairs is sampled independently from $\mathrm{Bern}(1/2)$. Therefore each $Y_i^{(t)} \sim \Gnhalf^{\otimes 2}$.}
\label{fig:interpolation-paths}
\end{figure}
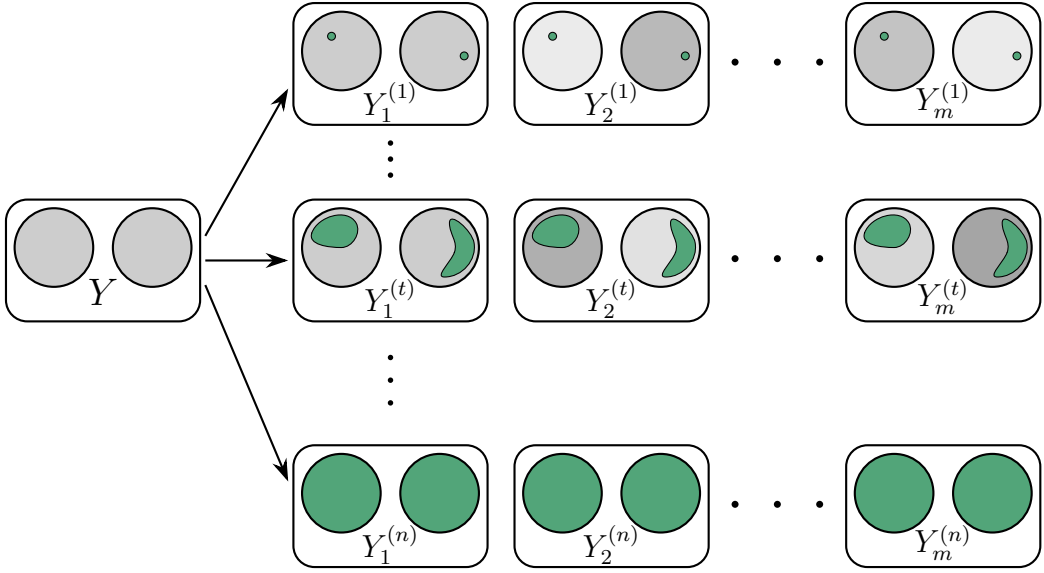

For any $m \in \N$ and $t \in [n]$, the random variable $X_{m,t}$ counts the number of sets of large solutions from the interpolation path $Y_1^{(t)}, \dots, Y_m^{(t)}$ that have mutual overlap of size at least $\lceil \gamma L\rceil$ in each subgraph, for 
\[
    \gamma := 2 - \frac \eps 2 
    \qquad 
    \text{and}
    \qquad 
    L := \log_2 n. 
\]
Intuitively, these sets should look like a pair of flowers, where each flower has $m$ petals and a central face corresponding to the overlapping portion --- see Figure~\ref{fig:flowers} and Definition~\ref{def:forbidden-structures}.

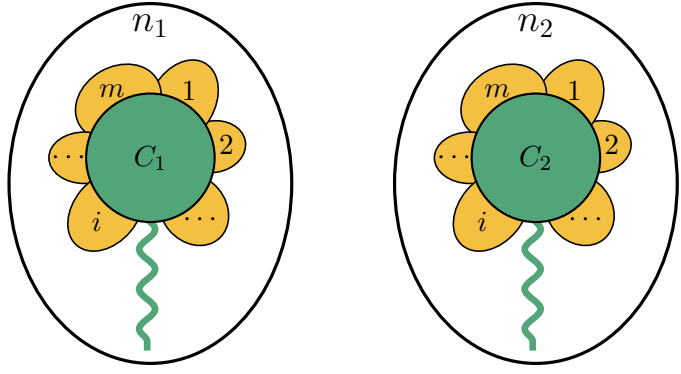
\begin{figure}
    \centering
        \begin{tikzpicture}[scale=0.85]
        
        \drawflower{-3}{0}[oval label=n_1,
            disc label=C_1, 
            disc fill=greenshape, 
            disc border=black, 
            petal fill=gray!20, 
            petal border=black, 
            petal 1 label=2, 
            petal 2 label=1, 
            petal 3 label=m, 
            petal 4 label=\cdots, 
            petal 5 label=i, 
            petal 6 label=\cdots,
            petal 1 fill=niceyellow,
            petal 2 fill=niceyellow,
            petal 3 fill=niceyellow,
            petal 4 fill=niceyellow,
            petal 5 fill=niceyellow,
            petal 6 fill=niceyellow,
        ]
        
        \drawflower{3}{0}[oval label=n_2,
            disc label=C_2, 
            disc fill=greenshape, 
            disc border=black, 
            petal fill=gray!20, 
            petal border=black, 
            petal 1 label=2, 
            petal 2 label=1, 
            petal 3 label=m, 
            petal 4 label=\cdots, 
            petal 5 label=i,
            petal 6 label=\cdots,
            petal 1 fill=niceyellow,
            petal 2 fill=niceyellow,
            petal 3 fill=niceyellow,
            petal 4 fill=niceyellow,
            petal 5 fill=niceyellow,
            petal 6 fill=niceyellow,
        ]
        \end{tikzpicture}
    \caption{$X_{m,t}$ counts the number of sets of large solutions $\{\sigma_i\}_{i \in [m]}$ to the $\MCS(n)$ problem with a particular overlap pattern; we visualize one such set here. Each $\sigma_i := \left(H_{i,1}, H_{i,2}\right)$ is a pair of isomorphic subgraphs of $Y_i^{(t)} = \left(G_{i,1}^{(t)}, G_{i,2}^{(t)}\right)$, 
    with 
    $|\sigma_{i}| \geq (2+\eps) \log_{2} n$. 
    To see the overlap pattern, we overlay each $V(H_{i,j})$ onto $n_j := [n]$, where the vertices $P_j^{(t)}$ processed by $\Alg$ through step $t$ are colored green. Each $V(H_{i,j})$ is the union of the yellow petal $i$ in $n_j$ and the green disc $C_j$, of size $|C_j| = \lceil(2 - \eps/2)\log_2 n\rceil$. Crucially, note that $C_j = V(H_{i,j}) \cap P_j^{(t)}$ for all~$i \in [m]$.}
    \label{fig:flowers}
\end{figure}

In Lemma~\ref{lem:S-above}, we show that
\begin{align*}
    \Pr(\findManyLarge) \leq \sum_{t=1}^n \Pr(X_{m,t} \geq 1).
\end{align*}
In fact, the event $\findManyLarge$ states that $\Alg$ is run on each input $Y_1^{(t)}, \dots, Y_m^{(t)}$, for a particular $t = \tau \in [n]$. The proof of Lemma~\ref{lem:S-above} relies on setting $\tau$ to be the first time step at which $\Alg(Y)$ finds a solution of size $\lceil \gamma L \rceil$ (or $\tau = n$ if such a solution is never found).

With this definition in hand, the proof of Lemma~\ref{lem:S-above} is relatively straightforward: $\findManyLarge$ is the event that $\Alg$ finds a large solution in each correlated input $Y_i^{(\tau)}$ and we have defined $Y_1^{(\tau)}, \dots, Y_m^{(\tau)}$ such that the first $\lceil \gamma L \rceil$ nodes in each solution are identical. Therefore the solutions have substantial mutual overlap, so $X_{m,\tau} \geq 1$. 

We call these $m$-tuples of overlapping solutions \emph{forbidden structures}, since we show in Lemma~\ref{lem:multi-OGP} that they exist with vanishing probability as $n \to \infty$:
\begin{align*}
    \sum_{t=1}^n \Pr(X_{m,t} \geq 1) \leq n^{-\log_2 n + o(\log n)}.
\end{align*}
In essence, Lemma~\ref{lem:multi-OGP} says that the $\MCS$ problem exhibits an overlap gap property --- it stipulates the non-existence of certain tuples of large solutions with a particular overlap pattern. Similar to other proofs of the overlap gap property, we prove Lemma~\ref{lem:multi-OGP} by the first moment method and analyzing $\E[X_{m,t}]$. To enable our analysis, we say that each $\sigma_i$ has size $\alpha_i L$, for some
\begin{align*}
    \alpha_i \geq 2 + \eps. %
\end{align*}
Throughout the proof we use the standard upper bound $\binom{n}{k} \leq n^k$. We can then upper bound the number of such forbidden structures by
\begin{align}
    \left[\binom{n}{\lceil \gamma L\rceil} \prod_{i=1}^m \binom{n}{(\alpha_i - \gamma)L}\right]^2 \leq n^{2L\left(\gamma + \sum_{i=1}^m (\alpha_i - \gamma)\right)}, \label{OGP-counting-informal}
\end{align}
since the first term counts the number of choices for the set of common vertices of size $\lceil \gamma L\rceil$ in each subgraph, and the second term counts the number of choices for the leftover $(\alpha_i - \gamma)L$ vertices in each subgraph of each solution $\sigma_i$.

To analyze the probability that a given forbidden structure is in fact a set of solutions to $\MCS(n)$, we consider each $\sigma_i = (H_{i,1}, H_{i,2})$, fix a bijection $\pi_i: V(H_{i,1}) \to V(H_{i,2})$, and calculate the probability that $\pi_i$ is a graph isomorphism, conditioned on the subgraph already seen by $\Alg$ through step $t$. Ultimately, we  take a union bound over all such permutations $\pi_i$, but since the subgraphs are of size $O(\log n)$, this contributes a multiplicative factor of only $(\log n)^{\log n} \approx n^{\log \log n} < n^{o(L)}$, which is negligible. 
Since each $Y_i^{(t)} \sim \Gnhalf^{\otimes 2}$, $\pi_i$ preserves adjacencies independently with probability $1/2$ for every pair of vertex pairs $(u, u') \in \binom{V(H_{i,1})}{2}$ and $(\pi_i(u), \pi_i(u')) \in \binom{V(H_{i,2})}{2}$, so long as the connectivity of at least one of the pairs has not been determined by our conditioning. There are at least $\binom{\alpha_i L}{2} - \binom{\lceil \gamma L\rceil}{2}$ such pairs (this uses that $t \leq n - \Omega(\log n)$; the case when $t$ is very close to $n$ is simple to handle separately), so since each $Y_i^{(t)}$ is conditionally independent, we can upper bound the probability that all $\sigma_i$ are solutions by
\begin{align}
    \prod_{i=1}^m 2^{-\left(\binom{\alpha_i L}{2} - \binom{\gamma L}{2}\right)} \approx 2^{-\frac{L^2}{2}\sum_{i=1}^m\left(\alpha_i^2 - \gamma^2\right)} = n^{-\frac{L}{2}\sum_{i=1}^m\left(\alpha_i^2 - \gamma^2\right)}, \label{OGP-probability-informal}
\end{align}
where the $\approx$ symbol means we have dropped lower-order terms in the exponent. Leveraging the linearity of expectation together with estimates~(\ref{OGP-counting-informal}) and~(\ref{OGP-probability-informal}), we get that
\begin{align*}
    \E[X_{m,t}] \leq n^{-\frac L 2\left(-4\gamma + \sum_{i=1}^m\left[(\alpha_i^2 - 4\alpha_i) - (\gamma^2 - 4\gamma)\right]\right)}. %
\end{align*}
Moreover, since $\alpha_i^2 - 4\alpha_i \geq \eps^2 - 4$ for all $\alpha_i \geq 2 + \eps$ and also $\gamma^2 - 4\gamma = \eps^2/4 -4$, we have that 
\begin{align*}
    -4\gamma + \sum_{i=1}^m\left[(\alpha_i^2 - 4\alpha_i) - (\gamma^2 - 4\gamma)\right] \geq -8 + m\left(\frac{3\eps^2}{4}\right).
\end{align*}
So by choosing $m := \left\lceil 40/(3\eps^2) \right\rceil$ we obtain $\E[X_{m,t}] \leq n^{-L}$, which proves Lemma~\ref{lem:multi-OGP}. Chaining together Lemmas~\ref{lem:S-below},~\ref{lem:S-above}, and~\ref{lem:multi-OGP}, we get that
\begin{align}
    \Pr(\findOneLarge)^m \leq \Pr(\findManyLarge) \leq \sum_{t=1}^{n} \Pr(X_{m,t} \geq 1) \leq n^{-\log_2 n + o(\log n)}
\end{align}
which immediately proves Theorem~\ref{thm:impossibility-formal}.

\section{Analysis of the greedy algorithm}\label{sec:greedy-proof}

In this section we analyze the \greedy algorithm described in Algorithm~\ref{alg:greedy} and prove Theorem~\ref{thm:algs-formal}. 

\begin{algorithm}[ht!]
\textbf{Input:} Two graphs $G_{1}$ and $G_{2}$. \\ 
\textbf{Output:} Two induced subgraphs $H_{1} \subseteq G_{1}$ and $H_{2} \subseteq G_{2}$ such that $H_{1} \cong H_{2}$, and a graph isomorphism $\pi : V(H_{1}) \to V(H_{2})$.
\begin{algorithmic}[1]
\Procedure{Greedy}{$G_1, G_2$} 
    \State Let $u_1, \dots, u_n$ be an arbitrary ordering of $V(G_{1})$. 
    \State Let $v_1, \dots, v_n$ be an arbitrary ordering of $V(G_{2})$.
    \State Initialize $S_1:= \emptyset$ and $ S_{2} := \emptyset$. 
    \State Initialize $\pi : S_{1} \to S_{2}$ to be the empty mapping.
    \For{$i \in [n]$}
        \State $\done := \F$;
        \State $j := 1$;
        \While{$\done == \F$ and $j \leq i$}
            \If{$v_{j} \notin S_{2}$ and $\textsc{IdenticalConnections}(u_{i}, v_{j}, S_{1}, S_{2}, \pi, G_{1}, G_{2})$}
                \State $S_{1} \gets S_{1} \cup \{ u_{i} \}$
                \State $S_{2} \gets S_{2} \cup \{ v_{j} \}$
                \State $\pi(u_{i}) := v_{j}$;
                \State $\done \gets \T$;
            \ElsIf{$u_{j} \notin S_{1}$ and $\textsc{IdenticalConnections}(u_{j}, v_{i}, S_{1}, S_{2}, \pi, G_{1}, G_{2})$}
                \State $S_{1} \gets S_{1} \cup \{ u_{j} \}$
                \State $S_{2} \gets S_{2} \cup \{ v_{i} \}$;
                \State $\pi(u_{j}) := v_{i}$;
                \State $\done \gets \T$;
            \Else 
                \State $j \gets j+1$;
            \EndIf
        \EndWhile 
    \EndFor
    \State $H_{1} := G_{1}[S_{1}]$;
    \State $H_{2} := G_{2}[S_{2}]$;
    \State \Return $(H_{1}, H_{2}, \pi)$
\EndProcedure
\State
\Procedure{IdenticalConnections}{$u, v, S_1, S_2, \pi, G_1, G_2$}
    \State \Return $\forall u' \in S_{1}: \quad (u, u') \in E(G_1) \iff (v, \pi(u')) \in E(G_2)$
\EndProcedure
\end{algorithmic}
\caption{Greedy algorithm for finding common induced subgraphs}\label{alg:greedy}
\end{algorithm}

\begin{proof}[Proof of Theorem~\ref{thm:algs-formal}]
Throughout the execution of \greedy the induced subgraphs 
$G_{1}[S_{1}]$ and~$G_{2}[S_{2}]$ 
are isomorphic by construction, 
with $\pi : S_{1} \to S_{2}$ being a graph isomorphism. 
Consequently, this remains true for the output of \greedy too. 

We now turn to the lower bound on the size of the output. 
We first introduce some helpful notation. 
Let $S_{1}(t)$ denote the set $S_{1}$ after \greedy processes the vertex pair $(u_{t}, v_{t})$, 
and define $S_{2}(t)$ analogously. 
For $k \in \N$, let $T_{k} := \min \left\{ t : |S_{1}(t)| = k \right\}$ denote the first time when $S_{1}$ has size~$k$. 
For convenience we set $T_{0} := 0$, note that $T_{1} = 1$, and by convention we set $T_{k} = \infty$ if $|S_{1}(n)| < k$. 
To abbreviate expressions that appear throughout the proof, 
set 
\[
\ell := \left\lceil (2-\eps) \log_2 n \right\rceil
\qquad 
\text{and} 
\qquad 
\widetilde{n} := \left\lfloor \frac{n}{2 \log_2 n} \right\rfloor,
\]
and note that $\ell \widetilde{n} \leq n$ (assuming $\eps \geq 1/\log_{2}(n)$).

With the notation above, the size of the output of \greedy is $|S_{1}(n)|$ 
and our goal is to show that $|S_{1}(n)| \geq (2-\eps) \log_2 n$ 
with high probability. 
Observe that $|S_{1}(n)| \geq (2-\eps) \log_2 n$ if and only if $T_{\ell} \leq n$. 
Note also that if $T_{\ell} > n$, 
then there exists $k \in \{1, \ldots, \ell - 1\}$ such that 
$T_{k} \leq k \widetilde{n}$ and $T_{k+1} > (k+1) \widetilde{n}$. Thus by a union bound we have that 
\begin{equation}\label{eq:error_blocks}
\Pr \left( |S_{1}(n)| < (2-\eps) \log_{2}n \right) 
= \Pr( T_{\ell} > n )
\leq \sum_{k=1}^{\ell-1} \Pr \left( T_{k} \leq k \widetilde{n}, T_{k+1} > (k+1) \widetilde{n} \right).
\end{equation}

In the remainder of the proof we fix $k \in \{1, \ldots, \ell - 1\}$
and bound the probability on the right hand side of~\eqref{eq:error_blocks}. 
To do so, we condition on the sets $S_{1}(k \widetilde{n})$ and $S_{2}(k \widetilde{n})$, obtaining that 
\[
\Pr \left( T_{k} \leq k \widetilde{n}, T_{k+1} > (k+1) \widetilde{n} \right) 
= \E \left[ \Pr \left( T_{k} \leq k \widetilde{n}, T_{k+1} > (k+1) \widetilde{n} \, \middle| \, S_{1}(k \widetilde{n}), S_{2}(k \widetilde{n}) \right)  \right].
\]
Note that 
$\Pr \left( T_{k} \leq k \widetilde{n}, T_{k+1} > (k+1) \widetilde{n} \, \middle| \, S_{1}(k \widetilde{n}), S_{2}(k \widetilde{n}) \right) = 0$ 
unless 
$|S_{1}(k \widetilde{n})| = |S_{2}(k \widetilde{n})| = k$. 
Thus in what follows we may and will assume that 
$|S_{1}(k \widetilde{n})| = |S_{2}(k \widetilde{n})| = k$. 

Let 
$U_{k} := \left\{ u_{k \widetilde{n}+1}, \ldots, u_{(k+1)\widetilde{n}} \right\}$ 
and 
$V_{k} := \left\{ v_{k \widetilde{n}+1}, \ldots, v_{(k+1)\widetilde{n}} \right\}$, 
and note that $|U_{k}| = |V_{k}| = \widetilde{n}$. 
For $u \in U_{k}$, 
let $X_{1,k}(u) \in \{0,1\}^{k}$ be the random variable encoding the connectivity of $u$ to the vertices in $S_{1}(k \widetilde{n})$; that is, the coordinates of $X_{1,k}(u)$ are the indicator random variables $\mathbf{1}{\{ (u,u') \in E(G_{1})\}}$ 
for $u' \in S_{1}(k \widetilde{n})$. 
For $v \in V_{k}$, 
define $X_{2,k}(v) \in \{0,1\}^{k}$ analogously, 
and in particular order the coordinates in a way such that 
if the $i$-th coordinate of $X_{1,k}(u)$ corresponds to the indicator random variable $\mathbf{1} {\{ (u,u') \in E(G_{1})\}}$ for some $u' \in S_{1}(k \widetilde{n})$, 
then the $i$-th coordinate of $X_{2,k}(v)$ corresponds to the indicator random variable $\mathbf{1} {\{ (v,\pi(u')) \in E(G_{2})\}}$. 
With these definitions, observe that 
if $T_{k} \leq k \widetilde{n}$ and $T_{k+1} > (k+1) \widetilde{n}$, 
then for every $u \in U_{k}$ and for every $v \in V_{k}$ 
we have that $X_{1,k}(u) \neq X_{2,k}(v)$ 
(since if it were the case that $X_{1,k}(u) = X_{2,k}(v)$ 
for some $u \in U_{k}$ and $v \in V_{k}$, 
then \greedy would add $u$ to $S_{1}$ and $v$ to $S_{2}$, 
which contradicts $T_{k+1} > (k+1) \widetilde{n}$). 
Since $\{ X_{1,k}(u) \}_{u \in U_{k}}$ and $\{ X_{2,k} (v) \}_{v \in V_{k}}$ are all independent of $S_{1}(k \widetilde{n})$ and $S_{2}(k \widetilde{n})$, we thus have that 
\begin{equation}\label{eq:block_gap}
\Pr \left( T_{k} \leq k \widetilde{n}, T_{k+1} > (k+1) \widetilde{n} \, \middle| \, S_{1}(k \widetilde{n}), S_{2}(k \widetilde{n}) \right) 
\leq 
\Pr \left( \forall u \in U_{k} \ \forall v \in V_{k} : X_{1,k}(u) \neq X_{2,k}(v) \right). 
\end{equation}

To bound the probability in the display above, we can first condition on $\{ X_{1,k}(u) \}_{u \in U_{k}}$ and then use that $\{ X_{2,k} (v) \}_{v \in V_{k}}$ are i.i.d.\ uniform to obtain that 
\begin{align}
&\Pr \left( \forall u \in U_{k} \ \forall v \in V_{k} : X_{1,k}(u) \neq X_{2,k}(v) \right) \notag \\
&\qquad \qquad \qquad \qquad = \E \left[ \Pr \left( \forall u \in U_{k} \ \forall v \in V_{k} : X_{1,k}(u) \neq X_{2,k}(v) \, \middle| \, \{ X_{1,k}(u) \}_{u \in U_{k}} \right) \right] \notag \\
&\qquad \qquad \qquad \qquad = \E \left[ \Pr \left( \forall u \in U_{k} : X_{1,k}(u) \neq X_{2,k}(v) \, \middle| \, \{ X_{1,k}(u) \}_{u \in U_{k}} \right)^{\widetilde{n}} \right] \notag \\
&\qquad \qquad \qquad \qquad = \E \left[ \left( 1 - \frac{\left| \left\{ X_{1,k}(u) : u \in U_{k} \right\} \right|}{2^{k}} \right)^{\widetilde{n}} \right], \label{eq:key_prob_calc}
\end{align}
where in the second-to-last line $v$ is an arbitrary vertex in $V_{k}$. 
In the rest of the proof we bound the quantity in~\eqref{eq:key_prob_calc} in two different ways, depending on whether $k$ is small or large. 

First, assume that $k \leq (1-\eps/3) \log_{2} n$. 
Note that the quantity in the expectation in~\eqref{eq:key_prob_calc} is at most $1$, and it is equal to $0$ precisely when 
$\left| \left\{ X_{1,k}(u) : u \in U_{k} \right\} \right| = 2^{k}$. 
Consequently, we can bound~\eqref{eq:key_prob_calc} by the probability that 
$\left| \left\{ X_{1,k}(u) : u \in U_{k} \right\} \right| < 2^{k}$, 
obtaining that 
\begin{align*}
\E \left[ \left( 1 - \frac{\left| \left\{ X_{1,k}(u) : u \in U_{k} \right\} \right|}{2^{k}} \right)^{\widetilde{n}} \right]
&\leq 
\Pr \left( \left| \left\{ X_{1,k}(u) : u \in U_{k} \right\} \right| < 2^{k} \right) \\
&= \Pr \left( \exists x \in \{0,1\}^{k} : \forall u \in U_{k} : X_{1,k}(u) \neq x \right) \\
&\leq 2^{k} \left( 1 - \frac{1}{2^{k}} \right)^{\widetilde{n}} 
\leq 2^{k} \exp \left( - \frac{\widetilde{n}}{2^{k}} \right) 
\leq n \cdot \exp \left( 1 - \frac{n^{\eps/3}}{2 \log_2 n} \right),
\end{align*}
where in the last line we used a union bound, 
that $\{ X_{1,k}(u) \}_{u \in U_{k}}$ are i.i.d.\ uniform, 
the inequality $1-x \leq e^{-x}$, 
and also that $2^{k} \leq n^{1-\eps/3}$.

Next, assume that 
$(1-\eps/3) \log_{2} n \leq k \leq (2-\eps) \log_{2} n$. 
In this case we argue that with high probability 
$\left| \left\{ X_{1,k}(u) : u \in U_{k} \right\} \right| \geq n^{1-\eps/2}$, 
which allows us to bound the expectation in~\eqref{eq:key_prob_calc}. 
First, we bound the probability that 
$\left| \left\{ X_{1,k}(u) : u \in U_{k} \right\} \right| < n^{1-\eps/2}$
using a simple union bound as follows: 
\begin{multline*}
\Pr \left( \left| \left\{ X_{1,k}(u) : u \in U_{k} \right\} \right| < n^{1-\eps/2} \right) 
\leq 
\Pr \left( \exists S \subseteq \{0,1\}^{k} : |S| = \left\lfloor n^{1-\eps/2} \right\rfloor, \forall u \in U_{k} : X_{1,k}(u) \in S \right) \\
\leq \binom{2^{k}}{\left\lfloor n^{1-\eps/2} \right\rfloor} \left( \frac{\left\lfloor n^{1-\eps/2} \right\rfloor}{2^{k}} \right)^{\widetilde{n}} 
\leq \left( 2^{k} \right)^{n^{1-\eps/2}} \left( n^{-\eps/6} \right)^{\widetilde{n}} 
\leq n^{2n^{1-\eps/2} - \eps \widetilde{n}/6},
\end{multline*}
where we used that 
$n^{1-\eps/3} \leq 2^{k} \leq n^{2}$. We thus have that 
\[
\E \left[ \left( 1 - \frac{\left| \left\{ X_{1,k}(u) : u \in U_{k} \right\} \right|}{2^{k}} \right)^{\widetilde{n}} \right] 
\leq 
\left( 1 - \frac{n^{1-\eps/2}}{2^{k}} \right)^{\widetilde{n}} + n^{2n^{1-\eps/2} - \eps \widetilde{n}/6},
\]
and using the bound $2^{k} \leq n^{2-\eps}$ and the inequality $1-x \leq e^{-x}$ we have that 
\[
\left( 1 - \frac{n^{1-\eps/2}}{2^{k}} \right)^{\widetilde{n}} 
\leq 
\left( 1 - \frac{1}{n^{1-\eps/2}} \right)^{\widetilde{n}} 
\leq \exp \left( - \frac{\widetilde{n}}{n^{1-\eps/2}} \right) 
\leq \exp \left( 1 - \frac{n^{\eps/2}}{2 \log_2 n} \right).
\]

Finally, putting together~\eqref{eq:error_blocks},~\eqref{eq:block_gap},~\eqref{eq:key_prob_calc}, and the bounds on~\eqref{eq:key_prob_calc} discussed above, we have that 
\begin{multline}\label{eq:greedy_final_error_prob}
\Pr \left( |S_{1}(n)| < (2-\eps) \log_{2}(n) \right) \\
\leq (2 \log_2 n) \cdot 
\max \left\{ n \cdot \exp \left( 1 - \frac{n^{\eps/3}}{2 \log_2 n} \right), 
\exp \left( 1 - \frac{n^{\eps/2}}{2 \log_2 n} \right) + n^{2n^{1-\eps/2} - \eps \widetilde{n}/6} \right\}
\end{multline}
which is at most $\exp\left( - n^{\eps/4} \right)$ 
whenever $\eps \geq \frac{13 \log \log n}{\log n}$ 
and $n$ is large enough. 
\end{proof}

In fact, the error probability obtained in~\eqref{eq:greedy_final_error_prob} vanishes as $n \to \infty$ 
whenever 
$\eps = \eps_{n} \geq \frac{C \log \log n}{\log n}$ for a constant $C > 6$. 
Noting that $6 / \log(2)\approx 8.656 < 9$, we thus obtain the following result as a corollary (noting that the constant in front of the second-order $\log \log n$ term has not been optimized). 
\begin{corollary}\label{cor:greedy}
    Let $(G_1, G_2) \sim \Gnhalf^{\otimes 2}$, define $\greedy$  as in Algorithm~\ref{alg:greedy}, 
    let $(H_1, H_2, \pi)$ denote the output of $\greedy(G_1, G_2)$. 
    Then $(H_1, H_2)$ is a solution to $\MCS(n)$ and $\pi$ is a graph isomorphism between $H_{1}$ and $H_{2}$. 
    Furthermore, the size of $(H_1, H_2)$ is at least $2 \log_{2} n - 9 \log \log n$ with high probability.
\end{corollary}

For completeness we include a brief analysis of the runtime of the \greedy algorithm. 

\begin{claim}[Runtime]\label{clm:runtime}
    Fix $n \in \N$ and let $G_{1}$ and $G_{2}$ be two graphs with $n$ vertices. 
    The \greedy algorithm (Algorithm~\ref{alg:greedy}) with input $(G_{1}, G_{2})$ runs in $O(n^{3})$ time. 
    In addition, if $(G_{1}, G_{2}) \sim \Gnhalf^{\otimes 2}$, 
    then with high probability $\greedy(G_{1}, G_{2})$ runs in $O(n^{2} \log n)$ time. 
\end{claim}
\begin{proof}
    The \textsf{for} loop goes over $i \in [n]$, which gives a factor of $n$. For any $i$, the body of the \textsf{while} loop is executed at most $i \leq n$ times. Finally, running the \textsc{IdenticalConnections} procedure takes $O(|S_{1}|)$ time, which gives a factor of $O(n)$, since $|S_{1}| \leq n$ in the worst case. Overall, these three factors combine to give a worst-case running time of $O(n^{3})$. 

    In the average-case setting with high probability the largest common subgraph has size $O(\log n)$. Therefore, with high probability we have $|S_{1}| = O(\log n)$ throughout the execution of the algorithm, so running the \textsc{IdenticalConnections} procedure takes $O(|S_{1}|) = O(\log n)$ time throughout. 
\end{proof}

\section{Proof of impossibility for online algorithms}\label{sec:impossibility-proof}
In this section we prove Theorem~\ref{thm:impossibility-formal}. 
To abbreviate notation, we write 
$Y := (G_{1}, G_{2}) \sim \Gnhalf^{\otimes 2}$
for the input pair of graphs. 
Note that Definition~\ref{def:online}---and hence also Theorem~\ref{thm:impossibility-formal}---allows randomized algorithms. However, a simple averaging-based argument shows that it suffices to consider deterministic algorithms. 
We state this in the following lemma.

\begin{lemma}\label{lem:exists-deterministic}
    For any randomized algorithm $\Alg$, 
    there exists a random seed $\omega^* \in \Omega$ such that the deterministic algorithm $\Alg^{*}(Y) := \Alg(Y, \omega^*)$ satisfies
    \begin{align}
        \Pr_{Y, \omega}(|\Alg(Y, \omega)| \geq (2 + \eps) \log_2 n) \leq \Pr_{Y}(|\Alg^{*}(Y)| \geq (2 + \eps)\log_2 n). \label{line:better-deterministic}
    \end{align}
\end{lemma}
\begin{proof}
    Using the law of total expectation, we have that
    \begin{align*}
        \Pr_{Y, \omega}(|\Alg(Y, \omega))| \geq (2 + \eps)\log_2 n) = \E_\omega[\Pr_Y(|\Alg(Y, \omega)| \geq (2 + \eps)\log_2 n)].
    \end{align*}
    By an averaging argument, there exists some $\omega^* \in \Omega$ such that $\Alg^{*}(Y) := \Alg(Y, \omega^*)$ satisfies~\eqref{line:better-deterministic}.
\end{proof}

Due to Lemma~\ref{lem:exists-deterministic}, 
in order to prove Theorem~\ref{thm:impossibility-formal}, it suffices to prove it for deterministic online algorithms. 
In the following we thus only consider deterministic online algorithms. 
In particular, this means that the only source of randomness in what follows is the input $Y$. 

We now formally define the interpolation paths and forbidden structures that will be key to our analysis. These definitions closely follow the analogous structures in \cite{GKW25}, 
though we take care to 
adapt them to the ``two-dimensional'' nature of the $\MCS(n)$ problem. 

\begin{definition}[Interpolation paths]\label{def:interpolation-paths}
    Fix a deterministic $\Alg \in \online(n)$, $m \in \N$, and sample $Y = (G_{1}, G_{2}) \sim \Gnhalf^{\otimes 2}$. For $t \in [n]$, the $t^\text{th}$ \emph{interpolation path} is denoted $Y_1^{(t)}, \dots, Y_m^{(t)}$, where $Y_i^{(t)} = \left(G_{i, 1}^{(t)}, G_{i, 2}^{(t)}\right)$ and each $G_{i,j}^{(t)}$ is a graph on vertex set $[n]$. First, set
    \begin{align}
        Y_1^{(t)} := Y \qquad \forall t \in [n]. \label{def:first-Y}
    \end{align}
    The remaining inputs are constructed based on the behavior of $\Alg$ on input $Y$. 
    Let $P_1^{(t)}$ and $P_2^{(t)}$ denote the vertices in $G_1$ and $G_2$ (respectively) processed by $\Alg$ on input $Y$ after the first $t$ steps. 
    For every $t \in [n]$, $i \in [m]$, and $j \in \{1, 2\}$, let $e_{i, j}^{(t)}(\{u, v\})$ denote the edge status of the pair $\{u, v\}$ in~$G_{i, j}^{(t)}$. 
    Then, for each $t \in [n]$, $i \in \{2, \ldots, m \}$, and $j \in \{1, 2\}$, we set
    \begin{align}
        e_{i,j}^{(t)}(\{u, v\}) &:= e_{1, j}^{(t)}(\{u, v\}) \qquad \forall \{u, v\} \in \binom{P_j^{(t)}}{2}, \label{def:same-edges} \\
        e_{i,j}^{(t)}(\{u, v\}) &\overset{\text{iid}}{\sim} \mathrm{Bern}(1/2) \qquad \forall \{u, v\} \notin \binom{P_j^{(t)}}{2}, \label{def:iid-edges}
    \end{align}
    where all of the random bits in~\eqref{def:iid-edges} are independent of everything else.
\end{definition}
See Figure~\ref{fig:interpolation-paths} for an illustration of these interpolation paths. 

\begin{remark}\label{rem:interpret-interpolation}
    Define the family of interpolation paths $\{Y_1^{(t)}, \dots, Y_m^{(t)}\}_{t \in [n]}$ as above. Let $\Gamma_i^{(t)} \in \{0,1\}^{2 \binom{t}{2}}$ encode the edge status of every pair of vertices that $\Alg$ has inspected after $t$ steps on input~$Y_i^{(t)}$. Then~\eqref{def:first-Y} and~\eqref{def:same-edges} ensure that
    \begin{align*}
        \Gamma^{(t)} := \Gamma_1^{(t)} = \cdots = \Gamma_m^{(t)}.
    \end{align*}
    Since $\Alg$ makes deterministic decisions based only on the information it has received by that point, we have that $\Alg^{(t)}(Y_i^{(t)}) = \Alg^{(t)}(Y_{i'}^{(t)})$ for all $i, i' \in [m]$. Moreover,~\eqref{def:iid-edges} ensures that the inputs $\{Y_i^{(t)}\}_{i \in [m]}$ are mutually independent (and identically distributed), conditioned on $\Gamma^{(t)}$.  Combining this with~\eqref{def:first-Y}, we have that marginally
    \begin{align*}
        Y_i^{(t)} \sim \Gnhalf^{\otimes 2} \qquad \forall i \in [m], t \in [n].
    \end{align*}
\end{remark}

\begin{definition}[Forbidden structures]\label{def:forbidden-structures}
    Fix a deterministic $\Alg \in \online(n), m \in \N$, and $\eps > 0$, 
    let $Y = (G_1, G_2) \sim \Gnhalf^{\otimes 2}$, 
    and consider a family of interpolation paths $\{Y_1^{(t)}, \dots, Y_m^{(t)}\}_{t \in [n]}$. 
    For any $t \in [n]$, let $X_{m,t}$ count the number of $m$-tuples $\left\{\sigma_1, \dots, \sigma_m\right\}$ such that for every $i \in [m]$ the following two properties hold.
    \begin{enumerate}
        \item \emph{(Large solutions.)} We have that $\sigma_i = \left(H_{i,1}, H_{i,2}\right)$ is a solution to the $\MCS(n)$ problem on instance $Y_i^{(t)}$, with \label{prop:large-sol}
        \begin{align*}
            |\sigma_i| \geq (2 + \eps) \log_2 n.
        \end{align*} 
        \item \emph{(Overlap structure.)} For $j \in \{1, 2\}$ and $i \in [m]$ we have that \label{prop:two}
        \begin{align*}
            V(H_{1,j}) \cap P_j^{(t)} &= V(H_{i,j}) \cap P_j^{(t)}, \\
            \left| V(H_{1,j}) \cap P_j^{(t)} \right| &= \lceil (2 - \eps/2) \log_2 n\rceil. %
        \end{align*}
    \end{enumerate}
\end{definition}
See Figure~\ref{fig:flowers} for an illustration of these forbidden structures. 
With these definitions we are now ready to prove Theorem~\ref{thm:impossibility-formal}. 

\begin{proof}[Proof of Theorem~\ref{thm:impossibility-formal}]
    Recall from Lemma~\ref{lem:exists-deterministic} that we may and hence will assume that the algorithm $\Alg$ is deterministic.
    Given $\eps > 0$, we set the constant $m$ to be 
    \begin{align}
        m := \left\lceil \frac{40}{3\eps^2}\right\rceil. \label{def:m}
    \end{align}
    For notational convenience, in all of the arguments that follow, let
    \[
        \gamma := 2 - \frac \eps 2 
        \qquad 
        \text{and}
        \qquad 
        L := \log_2 n. 
    \]
    Recall that the input to the algorithm $\Alg$ is $Y = (G_{1}, G_{2}) \sim \Gnhalf^{\otimes 2}$. Sample a family of interpolation paths $\{Y_1^{(t)}, \dots, Y_m^{(t)}\}_{t \in [n]}$ based on $\Alg, m$, and $Y$ according to Definition~\ref{def:interpolation-paths}. 
    Define the stopping time $\tau$ to be
    \begin{align}
        \tau := \min\left\{\min \left\{1 \leq t \leq n: \left| \Alg^{(t)} (Y) \right| = \lceil\gamma L\rceil \right\}, n\right\}. \label{def:tau}
    \end{align}
    In words, $\tau$ is the first time when the algorithm $\Alg$ on input $Y$ reaches a (partial) solution of size \emph{close} to half-optimal (specifically, of size at least $(2-\eps/2) \log_{2} n$).
    Next, define the events
    \begin{align*}
        \findOneLarge := \{|\Alg(Y)| \geq (2 + \eps)L\} %
        \qquad 
        \text{and}
        \qquad
        \findManyLarge := \bigcap_{i = 1}^m \left\{ \left| \Alg \left( Y_i^{(\tau)} \right) \right| \geq (2 + \eps)L \right\} .
    \end{align*}
    Observe that $\findOneLarge$ is the event that we care about: that algorithm $\Alg$ on input $Y$ finds a solution \emph{beyond} half-optimal (specifically, of size at least $(2+\eps) \log_{2} n$). Furthermore, $\findManyLarge$ is the event that this is true for all of the inputs along the $\tau$-th interpolation path, namely for $Y_{1}^{(\tau)}, \ldots, Y_{m}^{(\tau)}$. 
    
    Recall the definition of $X_{m,t}$ from Definition~\ref{def:forbidden-structures}. 
    The core of the proof lies in the following chain of inequalities:
    \begin{equation}\label{eq:key_inequalities}
        \Pr(\findOneLarge)^m 
        \leq \Pr(\findManyLarge) 
        \leq \sum_{t=1}^n \Pr(X_{m,t} \geq 1) 
        \leq n^{-\log_2 n + o(\log n)}.
    \end{equation}
    We defer the proofs of these inequalities to Lemmas~\ref{lem:S-below},~\ref{lem:S-above}, and~\ref{lem:multi-OGP}, respectively, below.
    Theorem~\ref{thm:impossibility-formal} then follows by taking an $m$-th root in~\eqref{eq:key_inequalities}, noting that $m$ is constant (depending only on~$\eps$). 
\end{proof}

The rest of this section is devoted to proving the inequalities in~\eqref{eq:key_inequalities}. 
Lemma~\ref{lem:S-below} is an analogue of \cite[Lemma~4.6]{GKW25}; we include its proof here for completeness.
\begin{lemma}\label{lem:S-below}
    Defining $m$, $\findOneLarge$, and $\findManyLarge$ as in the proof of Theorem~\ref{thm:impossibility-formal}, we have that $\Pr(\findOneLarge)^m \leq \Pr(\findManyLarge)$.
\end{lemma}
\begin{proof}
    As in Remark~\ref{rem:interpret-interpolation}, let $\Gamma^{(t)}$ be the random variable that encodes all the information processed by $\Alg$ through step $t$ on input $Y$; that is, $\Gamma^{(t)}$ encodes the edge status of every pair $\{u_1, v_1\} \in \binom{P_1^{(t)}}{2}$ in $G_{1}$ and $\{u_2, v_2\} \in \binom{P_2^{(t)}}{2}$ in $G_{2}$. Since $\tau \in [n]$ deterministically, by the law of total probability we can write $\Pr(\findOneLarge)$ as follows:
    \begin{align}
        \Pr(\findOneLarge) 
        &= \sum_{t = 1}^n \Pr(\findOneLarge \cap \{\tau = t\}) %
        = \sum_{t = 1}^n \E\left[\Pr\left(\findOneLarge \cap \{\tau = t\} \, \middle| \, \Gamma^{(t)}\right)\right] \notag \\
        &= \sum_{t = 1}^n \E\left[\mathbf{1}{\{\tau = t\}}\Pr\left(\findOneLarge \, \middle| \, \Gamma^{(t)}\right)\right] %
        = \E\left[\sum_{t = 1}^n \mathbf{1}{\{\tau = t\}}\Pr\left(\findOneLarge \, \middle| \, \Gamma^{(t)}\right)\right], \notag
    \end{align}
    where the third equality follows because the event $\{\tau = t\}$ is determined by $\Gamma^{(t)}$. Considering now the right-hand side of the desired inequality, $\Pr(\findManyLarge)$, the same reasoning gives that 
    \begin{align*}
        \Pr(\findManyLarge) = \E\left[\sum_{t = 1}^n \mathbf{1}{\{\tau = t\}} \Pr\left(\findOneLarge_{1,t} \cap \cdots \cap \findOneLarge_{m,t} \, \middle| \, \Gamma^{(t)}\right)\right],
    \end{align*}
    where $\findOneLarge_{i, t} := \{|\Alg(Y_i^{(t)})| \geq (2 + \eps) \log_2 n\}$. 
    
    Now we use an important property of the construction of interpolation paths in Definition~\ref{def:interpolation-paths}: all edges in $Y_{1}^{(t)}, \ldots, Y_{m}^{(t)}$ that are revealed after step $t$ are i.i.d. 
    In particular, this implies that conditioned on $\Gamma^{(t)}$, the events $\findOneLarge_{1, t}, \ldots, \findOneLarge_{m, t}$ are mutually independent (see Remark~\ref{rem:interpret-interpolation}). 
    Thus, 
    \begin{align*}
        \Pr\left(\findOneLarge_{1,t} \cap \cdots \cap \findOneLarge_{m,t} \, \middle| \, \Gamma^{(t)}\right) 
        = \prod_{i = 1}^m \Pr\left(\findOneLarge_{i,t} \, \middle| \, \Gamma^{(t)}\right) 
        = \Pr\left(\findOneLarge_{1, t} \, \middle| \, \Gamma^{(t)}\right)^m 
        = \Pr\left(\findOneLarge \, \middle| \, \Gamma^{(t)}\right)^m.
    \end{align*}
    Note that the random variable $\mathbf{1}{\{\tau = t\}}$ is nonzero for exactly one term in the sum over $t \in [n]$. As such we have that
    \begin{align*}
        \sum_{t = 1}^n \mathbf{1}{\{\tau = t\}} \Pr\left(\findOneLarge \, \middle| \, \Gamma^{(t)}\right)^m 
        = \left(\sum_{t = 1}^n \mathbf{1}{\{\tau = t\}} \Pr\left(\findOneLarge \, \middle| \, \Gamma^{(t)}\right)\right)^m.
    \end{align*}
    Finally, since $f(x) = x^m$ is convex for $m \geq 1$ and $x \geq 0$, we can apply Jensen's inequality to get
    \begin{equation*}
        \Pr(\findOneLarge)^m 
        = \left( \E\left[\sum_{t = 1}^n \mathbf{1}{\{\tau = t\}}\Pr\left(\findOneLarge \, \middle| \, \Gamma^{(t)}\right)\right] \right)^m 
        \leq \E\left[\left(\sum_{t = 1}^n \mathbf{1}{\{\tau = t\}}\Pr\left(\findOneLarge \, \middle| \,  \Gamma^{(t)}\right)\right)^m\right] 
        = \Pr(\findManyLarge). \qedhere
    \end{equation*}
\end{proof}

We now turn to the second inequality in~\eqref{eq:key_inequalities}. 
Lemma~\ref{lem:S-above} is an analogue of \cite[Lemma~4.8]{GKW25}; we include its proof here for completeness.
\begin{lemma}\label{lem:S-above}
    Define $\findManyLarge$ as in the proof of Theorem~\ref{thm:impossibility-formal} and $X_{m,t}$ as in Definition~\ref{def:forbidden-structures}. We have that
    \begin{align*}
        \Pr(\findManyLarge) \leq \sum_{t = 1}^n \Pr(X_{m,t} \geq 1).
    \end{align*}
\end{lemma}

\begin{proof}
    As in the proof of Lemma~\ref{lem:S-below}, the law of total probability gives that
    \begin{align*}
        \Pr(\findManyLarge) = \sum_{t = 1}^n \Pr\left(\findManyLarge \cap \{\tau = t\}\right),
    \end{align*}
    so it suffices to show that for every $t \in [n]$ the event $\findManyLarge \cap \{\tau = t\}$ implies $X_{m,t} \geq 1$. Fix any $t \in [n]$ and suppose that $\findManyLarge \cap \{\tau = t\}$ occurs. For each $i \in [m]$, define
    \begin{align*}
        \sigma_i := \Alg(Y_i^{(\tau)}).
    \end{align*}
    To see that $\{\sigma_1, \dots, \sigma_m\}$ is a forbidden structure and hence $X_{m,t} \geq 1$, note that Property~\ref{prop:large-sol} of Definition~\ref{def:forbidden-structures} is directly implied by $\findManyLarge$ and Property~\ref{prop:two} is implied by Remark~\ref{rem:interpret-interpolation} and $\{\tau = t\}$.
\end{proof}

Finally, we turn to proving the third inequality in~\eqref{eq:key_inequalities}. This establishes a variant of the OGP for the $\MCS(n)$ problem. 

\begin{lemma}\label{lem:multi-OGP}
    Define $X_{m,t}$ as in Definition~\ref{def:forbidden-structures}. We have that 
    \begin{align*}
        \sum_{t = 1}^n \Pr(X_{m,t} \geq 1) \leq n^{-\log_2 n + o(\log n)}.
    \end{align*}
\end{lemma}

We first deal with some simple edge cases in $t$. 
Note that if $t < \lceil (2 - \eps/2) \log_2 n\rceil$, then $X_{m,t} = 0$ deterministically. This is because then 
$|V(H_{1,1}) \cap P_{1}^{(t)}| \leq |P_{1}^{(t)}| = t < \lceil (2 - \eps/2) \log_2 n\rceil$, 
and so Property~\ref{prop:two} of Definition~\ref{def:forbidden-structures} cannot hold. 
Similarly, we claim that if 
\[
t > n - \left( \lceil (2 + \eps) \log_2 n\rceil - \lceil (2 - \eps/2) \log_2 n\rceil \right),
\]
then $X_{m,t} = 0$ deterministically. 
To see this, observe that if Property~\ref{prop:two} of Definition~\ref{def:forbidden-structures} holds, 
then the solution $\sigma_{i}$ at time $n$ can have size at most 
$\lceil (2 - \eps/2) \log_2 n\rceil + (n-t)$. 
Thus, in order for Property~\ref{prop:large-sol} to hold, 
we must have that 
$\lceil (2 - \eps/2) \log_2 n\rceil + (n-t) 
\geq (2+\eps) \log_{2} n$. 

Thus in the following we focus on times $t$ satisfying 
\begin{equation}\label{eq:non_edge_cases}
\lceil (2 - \eps/2) \log_2 n\rceil 
\leq t 
\leq n - \left( \lceil (2 + \eps) \log_2 n\rceil - \lceil (2 - \eps/2) \log_2 n\rceil \right).
\end{equation}
For $t \in [n]$ that does not satisfy~\eqref{eq:non_edge_cases}, the previous paragraph shows that 
$\Pr \left( X_{m,t} \geq 1 \right) = 0$. 

To analyze $X_{m,t}$, it is useful to use a truncation argument; 
to this end, we introduce the following quantities.
\begin{definition}\label{def:WZ}
Let $W_{m,t}$ count the number of $m$-tuples $\{\sigma_1, \dots, \sigma_m\}$ that satisfy Definition~\ref{def:forbidden-structures} and have $\max_{i \in [m]} |\sigma_i| > 6 \log_2 n$. 
Similarly, let 
$Z_{m,t}$ count the number of $m$-tuples $\{\sigma_1, \dots, \sigma_m\}$ that satisfy Definition~\ref{def:forbidden-structures} and have $\max_{i \in [m]} |\sigma_i| \leq 6 \log_2 n$. 
\end{definition}
By definition, we have that 
$X_{m,t} = W_{m,t} + Z_{m,t}$. 
Therefore, we have that
\begin{equation}\label{eq:triangle_ineq}
\Pr\left( X_{m,t} \geq 1 \right) 
\leq \Pr \left( W_{m,t} \geq 1 \right)  + \Pr \left( Z_{m,t} \geq 1 \right),
\end{equation}
and the following two claims bound these latter two quantities.

\begin{claim}\label{clm:big-tuples}
    Define $W_{m,t}$ as in Definition~\ref{def:WZ}. For every $t \in [n]$ we have that 
    \begin{align*}
        \Pr(W_{m,t} \geq 1) \leq m n^{-\log_2 n + o(\log n)}.
    \end{align*}
\end{claim}

\begin{claim}\label{clm:small-tuples}
    Define $Z_{m,t}$ as in Definition~\ref{def:WZ}. Suppose that $t$ satisfies~\eqref{eq:non_edge_cases}. Then we have that 
    \begin{align*}
        \Pr(Z_{m,t} \geq 1) \leq n^m n^{-\log_2 n + o(\log n)}.
    \end{align*}
\end{claim}

Lemma~\ref{lem:multi-OGP} follows directly from Claims~\ref{clm:big-tuples} and~\ref{clm:small-tuples}. 

\begin{proof}[Proof of Lemma~\ref{lem:multi-OGP}]
    By~\eqref{eq:triangle_ineq} and Claims~\ref{clm:big-tuples} and~\ref{clm:small-tuples}, we have that  
    \begin{align*}
        \sum_{t=1}^n \Pr(X_{m,t} \geq 1) &\leq \sum_{t=1}^n \Bigl(\Pr(W_{m,t} \geq 1) + \Pr(Z_{m,t} \geq 1)\Bigr) \notag 
        \leq \sum_{t=1}^n \Bigl(mn^{-\log_2 n + o(\log n)} + n^m n^{-\log_2 n + o(\log n)}\Bigr) \notag \\
        &\leq 2mn^{m+1} n^{-\log_2 n + o(\log n)} 
        \leq n^{-\log_2 n + o(\log n)},
    \end{align*}
    where we use that $a + b \leq 2ab$ for $a, b \geq 1$, and that $m$ is a constant, so $2m n^{m+1} \leq n^{o(\log n)}$.
\end{proof}

It remains to prove the claims. Claim~\ref{clm:big-tuples} follows from a straightforward union bound.

\begin{proof}[Proof of Claim~\ref{clm:big-tuples}]
    Note that $W_{m,t} \geq 1$ implies the existence of some $Y_i^{(t)} = \left(G_{i,1}^{(t)}, G_{i,2}^{(t)}\right)$ that contains a pair of subgraphs $H_1 \subseteq G_{i,1}^{(t)}, H_2 \subseteq G_{i,2}^{(t)}$ such that $H_1 \cong H_2$ and
    \begin{align*}
        k := |H_1| = |H_2| = \lceil 6\log_2 n \rceil.
    \end{align*}
    For any two subgraphs $H_1, H_2 \sim \mathbb{G}(k, 1/2)$, the probability that $H_1 \cong H_2$ is at most $(k!)2^{-\binom{k}{2}}$. Taking a union bound over all possible choices of $i \in [m]$ and subgraphs $H_1, H_2$, we see that
    \begin{align*}
        \Pr(W_{m,t} \geq 1) &\leq \sum_{i=1}^m \sum_{\substack{V(H_1) \subseteq [n], \, V(H_2) \subseteq [n] \\ |V(H_1)| = |V(H_2)| = k}} (k!) 2^{-\binom{k}{2}} 
        \leq m \binom{n}{k}^2 (k!) 2^{- \binom{k}{2}} \\
        &\leq m \left(n^22^{-(k-1)/2}\right)^k  
        \leq m \left(\sqrt{2} n^{-1}\right)^k 
        \leq m n^{-6\log_2 n + o(\log n)},
    \end{align*}
    where the third inequality uses that 
    $\binom{n}{k}^{2} k! \leq n^{2k}$, 
    the fourth inequality uses that $k \geq 6 \log_{2} n$ and hence $2^{-(k-1)/2} \leq \sqrt{2} n^{-3}$, 
    and the final inequality uses that $k \geq 6 \log_{2} n$ as well. 
\end{proof}

The proof of Claim~\ref{clm:small-tuples} is more involved. We use the following definition.
\begin{definition}\label{def:auxiliary-Z}
    For any $m \in \N$ and $\avec = (\alpha_1, \dots, \alpha_m) \in [2 + \eps, 6]^m$, let $Z_{\avec, m, t}$ count the number of $m$-tuples $\{\sigma_1, \dots, \sigma_m\}$ that satisfy Definition~\ref{def:forbidden-structures} and have $|\sigma_i| = \alpha_i \log_2 n$ for all $i \in [m]$.
\end{definition}
Note that we only consider $\avec$ such that $\avec \log_2 n \in [n]^m$. 
By definition, we have that 
\[
Z_{m,t} = \sum_{\substack{\avec \in [2+\eps, 6]^m \\ \avec \log_2 n \in [n]^m}} Z_{\avec, m, t},
\]
and so by a union bound we have that 
\[
\Pr(Z_{m,t} \geq 1) 
\leq \sum_{\substack{\avec \in [2+\eps, 6]^m \\ \avec \log_2 n \in [n]^m}} \Pr \left( Z_{\avec, m, t} \geq 1 \right).
\] 
Let $\alpha_{\max} := \max_{i \in [m]} \alpha_{i}$ 
and recall that $L = \log_{2} n$. 
Observe that if $t$ is such that 
$n-t < \alpha_{\max} L - \left\lceil \gamma L \right\rceil$, 
then 
$Z_{\avec, m, t} = 0$ deterministically. 
To see this, let $i^{*}$ be such that $\alpha_{i^{*}} = \alpha_{\max}$. 
If Property~\ref{prop:two} of Definition~\ref{def:forbidden-structures} holds, 
then the solution $\sigma_{i^{*}}$ at time $n$ can have size at most 
$\left\lceil \gamma L \right\rceil + (n-t)$, 
so in order for the solution $\sigma_{i^{*}}$ at time $n$ to have size $\alpha_{i^{*}} L = \alpha_{\max} L$, 
we must have that 
$\left\lceil \gamma L \right\rceil + (n-t) \geq \alpha_{\max} L$.
Thus when analyzing $Z_{\avec, m, t}$, we may assume that $t$ satisfies 
$\left\lceil \gamma L \right\rceil + (n-t) \geq \alpha_{\max} L$, 
or in other words that 
\begin{equation}\label{eq:t_condition_alphamax}
t \leq n - (\alpha_{\max} L - \left\lceil \gamma L \right\rceil).
\end{equation}
Applying Markov's inequality thus yields that 
\begin{equation}\label{eq:Z_bound}
    \Pr(Z_{m,t} \geq 1) 
    \leq \sum_{\substack{\avec \in [2+\eps, 6]^m, \, \avec L \in [n]^m \\ t \leq n - (\alpha_{\max} L - \left\lceil \gamma L \right\rceil)}} \Pr \left( Z_{\avec, m, t} \geq 1 \right)
    \leq \sum_{\substack{\avec \in [2+\eps, 6]^m, \, \avec L \in [n]^m \\ t \leq n - (\alpha_{\max} L - \left\lceil \gamma L \right\rceil)}} \E \left[Z_{\avec, m, t} \right].
\end{equation}
We proceed to bound $\E \left[ Z_{\avec, m, t} \right]$ for $t$ such that~\eqref{eq:t_condition_alphamax} holds. 
Lemma~\ref{lem:Z-counting} and Lemma~\ref{lem:Z-probability} handle the counting and the probability estimates, respectively.

\begin{lemma}\label{lem:Z-counting}
    Recall that $\gamma = 2 - \eps/2$ and $L = \log_2 n $. For any $m \in \N$, $\avec \in [2+\eps, 6]^m$, and $\left(P_1^{(t)}, P_2^{(t)}\right) \subseteq [n]^2$, the total number of $m$-tuples $\{\sigma_1, \dots, \sigma_m\}$ that satisfy Property~\ref{prop:two} in Definition~\ref{def:forbidden-structures} and have $|\sigma_i| = \alpha_i \log_2 n$ for all $i \in [m]$ is at most
    \begin{align*}
        n^{2L \left[\gamma + \sum_{i=1}^m(\alpha_i - \gamma)\right]}.
    \end{align*}
\end{lemma}

\begin{proof}
    Recall that $\sigma_i = (H_{i,1}, H_{i,2})$. 
    By Property~\ref{prop:two} in Definition~\ref{def:forbidden-structures}, 
    for each $i \in [m]$ and $j \in \{1,2\}$ 
    we have that $V(H_{i,j}) \cap P_j^{(t)} = V(H_{1,j}) \cap P_j^{(t)}$ and $|V(H_{1,j}) \cap P_j^{(t)}| = \lceil \gamma L\rceil$. To bound the number of $m$-tuples counted in the lemma statement, we pick any $j \in \{1, 2\}$ and first note that the number of ways to choose the set of common vertices $V(H_{1,j}) \cap P_j^{(t)}$ is at most
    \begin{align*}
        \binom{ \left| P_j^{(t)} \right|}{\lceil \gamma L \rceil} 
        \leq \binom{n}{\lceil \gamma L \rceil} 
        \leq n^{\lceil \gamma L \rceil}.
    \end{align*}
    Since the size of each $H_{i,j}$ is $\alpha_i L$, the number of ways to choose the remaining vertices in each $H_{i, j}$ is at most
    \begin{align*}
        \prod_{i=1}^m \binom{n}{\alpha_i L - \lceil \gamma L \rceil} 
        \leq \prod_{i=1}^m n^{\alpha_i L - \lceil \gamma L \rceil} 
        = n^{\sum_{i=1}^m (\alpha_i L - \lceil \gamma L \rceil)}.
    \end{align*}
    As such, counting each $j \in \{1,2\}$ separately, we can upper bound the count by
    \begin{align*}
        n^{2\lceil \gamma L \rceil + 2\sum_{i=1}^m(\alpha_iL - \lceil \gamma L \rceil)} \leq n^{2L[\gamma + \sum_{i=1}^m(\alpha_i - \gamma)]}
    \end{align*}
    and the result follows.
\end{proof}

\begin{lemma}\label{lem:Z-probability}
    Recall that $\gamma = 2 - \eps/2$ and $L = \log_2 n$. 
    Let $m \in \N$ and $\avec \in [2+\eps, 6]^m$, and define 
    \begin{align*}
        \Phi := \sum_{i=1}^m \left(\frac{\alpha_i^2 - \gamma^2}{2}\right).
    \end{align*}    
    Let $t$ be such that~\eqref{eq:t_condition_alphamax} holds. 
    Let $\Gamma^{(t)}$ be the random variable that encodes all the information processed by $\Alg$ through step $t$. 
    Then, for any 
    $m$-tuple $\{\sigma_1, \dots, \sigma_m\}$ that satisfies 
    Property~\ref{prop:two} in Definition~\ref{def:forbidden-structures} 
    and has $|\sigma_i| = \alpha_iL$ for all $i \in [m]$, 
    we have that
    \begin{align*}
        \Pr\left(\sigma_i \text{ is a common induced subgraph of } Y_{i}^{(t)} \text{ for all } i \in [m] \, \middle| \, \Gamma^{(t)}\right) 
        \leq n^{-\Phi L + C m \log L}
    \end{align*}
    for an absolute constant $C < \infty$.
\end{lemma}
\begin{proof}
    First note that, as explained after Definition~\ref{def:auxiliary-Z}, the condition that $t$ satisfies~\eqref{eq:t_condition_alphamax} is necessary in order for there to exist an $m$-tuple $\{\sigma_1, \dots, \sigma_m\}$ that satisfies 
    Property~\ref{prop:two} in Definition~\ref{def:forbidden-structures} 
    and has $|\sigma_i| = \alpha_iL$ for all $i \in [m]$. 
    (Otherwise, if $t$ is such that~\eqref{eq:t_condition_alphamax} does not hold, then there is no such $m$-tuple $\{\sigma_1, \dots, \sigma_m\}$, and hence the claim is vacuously true.)

    For $j \in \{1,2\}$, define the set of common vertices
    \begin{align*}
        C_j := V(H_{1,j}) \cap P_j^{(t)}.
    \end{align*}
    Since $\{\sigma_1, \dots, \sigma_m\}$ satisfies Property~\ref{prop:two} from Definition~\ref{def:forbidden-structures}, we have that $|C_j| = \lceil \gamma L \rceil$ and $C_j \subseteq V(H_{i,j})$ for all $i \in [m]$. Consider a bijection $\pi_i: V(H_{i,1}) \to V(H_{i,2})$, for any $i \in [m]$. We will analyze the probability that $\pi_i$ is a graph isomorphism, conditioned on $\Gamma^{(t)}$. Define
    \begin{align*}
        M_i := \{u \in C_1: \pi_i(u) \in C_2\}
    \end{align*}
    to be the set of common vertices in $H_{i,1}$ that get mapped by $\pi_i$ to common vertices in $H_{i,2}$. Note that $|M_i| \leq \lceil \gamma L\rceil$. Furthermore, define $I_{i}(u,u')$ to be the indicator that $\pi_i$ maintains the edge status of $(u, u')$, that is,
    \begin{align*}
        I_{i}(u,u') := \mathbf{1}{\left\{e_{i,1}^{(n)}(\{u,u'\}) = e_{i,2}^{(n)}(\{\pi_i(u), \pi_i(u')\}) \right\}}.
    \end{align*}
    Observe that $\pi_{i}$ is a graph isomorphism if and only if $I_{i}(u,u') = 1$ for all $(u,u') \in \binom{V(H_{i,1})}{2}$.     
    Recall that $\Gamma^{(t)}$ determines the edge statuses of all $(u,u') \in \binom{M_{i}}{2}$ and their respective pairs $(\pi_i(u), \pi_i(u'))$, so for such pairs we have that $\Pr\left(I_{i}(u, u') = 1 \, \middle| \, \Gamma^{(t)}\right) \in \{0,1\}$. 
    In particular, if $\Gamma^{(t)}$ is such that $I_{i}(u,u') = 0$ for some $(u,u') \in \binom{M_{i}}{2}$, 
    then 
    $\Pr\left(\pi_i \text{ is a graph isomorphism} \, \middle| \, \Gamma^{(t)}\right) = 0$. 
    However, in the argument that follows we will not use anything about $\Gamma^{(t)}$ --- merely that $t$ satisfies~\eqref{eq:t_condition_alphamax}. 
    Observe that if 
    $(u,u') \in \binom{V(H_{i,1})}{2} \setminus \binom{M_{i}}{2}$, 
    then either $(u,u') \notin \binom{P_{1}^{(t)}}{2}$ 
    or $(\pi_{i}(u), \pi_{i}(u')) \notin \binom{P_{2}^{(t)}}{2}$ 
    (or both). 
    Therefore,~\eqref{def:iid-edges} in Definition~\ref{def:interpolation-paths} ensures that 
    $\left\{ I_{i}(u,u') : (u,u') \in \binom{V(H_{i,1})}{2} \setminus \binom{M_{i}}{2} \right\}$ 
    are i.i.d.\ $\mathrm{Bern}(1/2)$ random variables. 
    Note further that\footnote{The equalities and inequalities in this display implicitly use that $t$ satisfies~\eqref{eq:t_condition_alphamax}, which guarantees the existence of an $m$-tuple $\{\sigma_{1}, \ldots, \sigma_{m}\}$ with appropriate properties, in particular that $|\sigma_{i}| = \alpha_{i} L$ and $|M_{i}| \leq \left\lceil \gamma L \right\rceil$.}  
    \[
    \left| \binom{V(H_{i,1})}{2} \setminus \binom{M_{i}}{2} \right| 
    = \binom{\alpha_i L}{2} - \binom{|M_i|}{2}
    \geq \binom{\alpha_i L}{2} - \binom{\lceil \gamma L \rceil}{2}
    \geq \frac{\alpha_{i}^{2} - \gamma^{2}}{2} L^{2} - 5 L.
    \]
    Thus we have that 
    \begin{align*}
        \Pr\left(\pi_i \text{ is a graph isomorphism} \, \middle| \, \Gamma^{(t)}\right) 
        &\leq \prod_{(u, u') \in \binom{V(H_{i,1})}{2} \setminus \binom{M_{i}}{2}} \Pr\left(I_{i}(u,u') = 1 \, \middle| \, \Gamma^{(t)}\right) \\
        &= 2^{-\left( \binom{\alpha_i L}{2} - \binom{|M_i|}{2}\right)} 
        \leq n^{-L(\alpha_i^2 - \gamma^2)/2 + 5}.
    \end{align*}
    Remark~\ref{rem:interpret-interpolation} implies that the event $H_{i,1} \cong H_{i,2}$ is independent for each $i \in [m]$, conditioned on $\Gamma^{(t)}$. Therefore, taking a union bound over all bijections $\pi_i$, we have that
    \begin{multline*}
        \Pr\left(\sigma_i \text{ is a common induced subgraph of } Y_{i}^{(t)} \text{ for all } i \in [m] \, \middle| \, \Gamma^{(t)}\right) \\
        \begin{aligned}
        &\leq \prod_{i=1}^m \sum_{\pi_i} \Pr\left(\pi_i \text{ is an isomorphism} \, \middle| \, \Gamma^{(t)}\right) 
        \leq \prod_{i=1}^m n^{-L(\alpha_i^2 - \gamma^2)/2 + 5} (\alpha_i L)! \\
        &\leq \prod_{i=1}^m n^{-L(\alpha_i^2 - \gamma^2)/2 + C \log L} 
        = n^{-L \left(\sum_{i=1}^m(\alpha_i^2 - \gamma^2)/2\right) + C m \log L}
        \end{aligned}
    \end{multline*}
for an absolute constant $C < \infty$, 
where we used that 
$(\alpha_{i} L)! \leq (6L)! \leq (6L)^{6L} \leq n^{C \log L}$.
\end{proof}

\begin{proof}[Proof of Claim~\ref{clm:small-tuples}]
    Recall that $\gamma = 2 - \eps/2$ and $L = \log_2 n$. 
    For $\avec \in [2 + \eps, 6]^m$, define 
    \[
    f(\alpha_i, \gamma) := \left(\frac{\alpha_i^2}{2} - 2 \alpha_i\right) - \left(\frac{\gamma^2}{2} - 2\gamma\right)
    \]
    and 
    \[
    \Psi(\avec) := -2\gamma +\sum_{i = 1}^m f(\alpha_i, \gamma). 
    \]

    For any $\avec \in [2 + \eps, 6]^m$ with $\avec L \in [n]^m$ and $t$ satisfying~\eqref{eq:t_condition_alphamax}, Lemmas~\ref{lem:Z-counting} and~\ref{lem:Z-probability} imply that
    \begin{align*}
        \E \left[ Z_{\avec, m, t} \right] 
        &= \E\left[ \E \left[ Z_{\avec, m, t} \, \middle| \, \Gamma^{(t)} \right] \right] 
        \leq \E \left[ n^{2L \left[\gamma + \sum_{i=1}^m(\alpha_i - \gamma)\right]} \cdot n^{-\Phi L + C m \log L} \right] \\
        &= n^{-L \Psi(\avec) + C m \log L}
        = n^{-L \Psi(\avec) + o(L)}.
    \end{align*}
    Plugging this estimate into~\eqref{eq:Z_bound}, we get that 
    \[
    \Pr(Z_{m,t} \geq 1) 
    \leq \sum_{\substack{\avec \in [2+\eps, 6]^m, \, \avec L \in [n]^m \\ t \leq n - (\alpha_{\max} L - \left\lceil \gamma L \right\rceil)}} \E \left[Z_{\avec, m, t} \right] 
    \leq n^{m} \cdot n^{-L \Psi(\avec) + o(L)},
    \]
    so what remains is to show that $\Psi(\avec) \geq 1$ for $\avec \in [2 + \eps, 6]^m$. 
    
    We first note that
    \begin{align*}
        \gamma^2 / 2 - 2\gamma = \eps^{2} / 8 - 2.
    \end{align*}
    Moreover, since for $\alpha_i \geq 2 + \eps$ the expression $\alpha_i^2/2 - 2\alpha_i$ is minimized at $\alpha_i = 2 + \eps$, we have that
    \begin{align*}
        \min_{\alpha_i \in [2 + \eps, 6]} \left\{\alpha_i^2/2 - 2\alpha_i \right\} = (2+ \eps)^2/2 - 2(2 + \eps) = \eps^2/2 - 2.
    \end{align*}
    Putting the previous two displays together, we thus have that 
    \begin{align*}
        f(\alpha_i, \gamma) \geq \left(\frac{\eps^2}{2} - 2\right) - \left(\frac{\eps^2}{8} - 2\right) = \frac{3\eps^2}{8}.
    \end{align*}
    By our choice of $m \geq 40/(3\eps^2)$ in~\eqref{def:m}, we have that 
    \begin{align*}
        \Psi(\avec) &= -2\gamma + \sum_{i=1}^m f(\alpha_i, \gamma) \geq -4 + m\left(\frac{3\eps^2}{8}\right) \geq -4 + \left(\frac{40}{3\eps^2}\right)\left(\frac{3\eps^2}{8}\right) = 1,
    \end{align*}
    which concludes the proof. 
\end{proof}

\section*{Acknowledgements}

This project was supported in part by the NSF-funded Institute for Data, Econometrics, Algorithms and Learning (IDEAL) through the grants NSF ECCS-2216970 and ECCS-2216912. We particularly thank the IDEAL Institute and Varun Gupta for developing and running the IDEAL Summer Research Exchange, where this work started in the summer of 2025. 
We are grateful to Vaidehi Srinivas for insightful and inspiring discussions concerning average-case reductions between large cliques and large common induced subgraphs.  
We also thank anonymous reviewers for helpful feedback.

\printbibliography

\end{document}